\documentclass[11pt]{article}   % For Latex2e
\usepackage{amsthm,amsmath}    % For Latex2e
\usepackage[dvips]{color}
\usepackage{amssymb}
\usepackage{epsfig}
\usepackage{latexsym}

\theoremstyle{plain}
\newtheorem{theorem}{Theorem}[section]
\newtheorem{lemma}[theorem]{Lemma}

\newtheorem{corollary}[theorem]{Corollary}

\theoremstyle{construction}
\newtheorem{construction}[theorem]{Construction}
\theoremstyle{definition}
\newtheorem{definition}[theorem]{Definition}
\newtheorem{example}[theorem]{Example}

\def\keywords{%
\list{}{\advance\topsep by0.35cm\relax\small\rm
 \leftmargin=1cm
 \itemindent\listparindent
 \rightmargin\leftmargin}\item[\hskip\labelsep\bf Keywords: ]}

\theoremstyle{remark}

\renewenvironment{proof}{\noindent{\it Proof}.}{\qed}

\title{Around Context-Free Grammars - a Normal Form, a Representation Theorem, and a Regular Approximation}  

\author{Liliana Cojocaru\\
School of Information Sciences, Computer Science\\ 
University of Tampere\\ 
Liliana.Cojocaru{@}uta.fi} 
\date{}

\begin{document}
\thispagestyle{empty}
\maketitle
\thispagestyle{empty}
\pagenumbering{arabic}

\begin{abstract}
We introduce a normal form for \textit{context-free grammars}, called \textit{Dyck normal form}. This 
is a syntactical restriction of the \textit{Chomsky normal form}, in which the two nonterminals occurring on 
the right-hand side of a rule are paired nonterminals. This pairwise property allows to define a homomorphism 
from Dyck words to words generated by a grammar in Dyck normal form. We prove that for each context-free language $L$, 
there exist an integer $K$ and a homomorphism $\varphi$ such that $L=\varphi(D'_K)$, where $D'_K\subseteq D_K$, 
and $D_K$ is the one-sided Dyck language over $K$ letters. Through a \textit{transition-like diagram} for a context-free 
grammar in Dyck normal form, we effectively build a regular language $R$ such that $D'_K=R\cap D_K$, which leads 
to the \textit{Chomsky-Sch$\ddot{u}$tzenberger theorem}. Using graphical approaches we refine $R$ such that 
the Chomsky-Sch\"utzenberger theorem still holds. Based on this readjustment we sketch a {\it transition diagram} 
for a {\it regular grammar} that generates a \textit{regular superset approximation} for the initial context-free language. 
\end{abstract}
\vspace*{-0.7cm}

\section*{Introduction}

A \textit{normal form} for context-free grammars consists of restrictions imposed on the structure of 
grammar's productions. These restrictions concern the number of terminals and nonterminals allowed on the 
right-hand sides of the rules, or on the manner in which terminals and nonterminals are arranged into the 
rules. Normal forms turned out to be useful tools in studying syntactical properties of context-free  grammars, in 
parsing theory, structural and descriptional complexity, inference and learning theory. Various normal forms 
for context-free  grammars have been study so far, but the most important remain the Chomsky normal form \cite{H}, 
Greibach normal form \cite{Gr}, and operator normal form \cite{H}. For definitions, results, and surveys on 
normal forms the reader is referred to \cite{CG}, \cite{H}, and \cite{PL}. A normal form is correct if it 
preserves the language generated by the original grammar. This condition is called {\it the weak equivalence}, 
i.e., a normal form preserves the language but may lose important syntactical or semantical properties of the 
original grammar. The more syntactical, semantical, or ambiguity properties a normal form preserves, the stronger 
it is. It is well known that the Chomsky normal form is a {\it strong} normal form.  

This paper is partly devoted to a new normal form for context-free  grammars, called \textit{Dyck normal form}. The Dyck normal 
form is a syntactical restriction of the Chomsky normal form, in which the two nonterminals occurring on the right-hand side 
of a rule are paired nonterminals, in the sense that each left (right) nonterminal of a pair has a unique right (left) 
pairwise. This pairwise property imposed on the structure of the right-hand side of each rule induces a nested 
structure on the derivation tree of each word generated by a grammar in Dyck normal form. More precisely, each derivation 
tree of a word generated by a grammar in Dyck normal form, that is read in the depth-first search order is a Dyck word, hence 
the name of the normal form. Furthermore, there exists always a homomorphism between the derivation tree of 
a word generated by a grammar in Chomsky normal form and its equivalent in Dyck normal form. In other words the Chomsky and Dyck normal forms are 
{\it strongly equivalent}. This property, along with several other terminal rewriting conditions imposed to a grammar in Dyck normal form, enable us to 
define a homomorphism from Dyck words to words generated by a grammar in Dyck normal form. We have been inspired to develop 
this normal form by the general theory of Dyck words and Dyck languages, that turned out to play a crucial role in 
the description and characterization of context-free  languages \cite{EHR2}, \cite{EHR3}, and \cite{IJR}. The definition and 
several properties of grammars in Dyck normal form are presented in Section 1. 

For each context-free  grammar $G$ in Dyck normal form we define, in Section 2, the \textit{trace language} associated with derivations 
in $G$, which is the set of all derivation trees of $G$ read in the depth-first search order, starting from the 
grammar axiom. By exploiting the Dyck normal form, and several characterizations of Dyck languages presented in \cite{IJR}, we 
give a new characterization of context-free  languages in terms of Dyck languages. We prove (also in Section 2) that for each 
context-free  language $L$, generated by a grammar $G$ in Dyck normal form, there exist an integer $K$ and a homomorphism $\varphi$ such 
that $L=\varphi(D'_K)$, where $D'_K$ (a subset of the Dyck language over $K$ letters) equals, with very little 
exceptions, the trace language associated with $G$. 

In Section 3 we show that the representation theorem in Section 2 emerges, through a \textit{transition-like diagram} 
for context-free  grammars in Dyck normal form, to the  Chomsky-Sch\"utzenberger theorem. By improving 
this transition diagram, in Section 4 we refine the regular language provided by the Chomsky-Sch\"utzenberger theorem, while in Section 5 
we show that the refined graphical representation of derivations in a context-free  grammar in Dyck normal form, used in the previous sections, provides 
a framework for a regular grammar that generates a \textit{regular superset approximation} for the initial context-free  language. 

The method used throughout this paper is graph-constructive, in the sense that it supplies a graphical interpretation 
of the Chomsky-Sch\"utzenberger theorem, and consequently it shows how to graphically build a regular language (as minimal as possible) 
that satisfies this theorem. Even if we reach the same famous Chomsky-Sch\"utzenberger theorem, the method used to approach it is 
different from the other methods known in the literature. In brief, the method in \cite{H} is based on pushdown approaches, 
while that in \cite{G} uses some kind of imaginary brackets that simulate the work of a pushdown store, when deriving a context-free  
language. The method presented in \cite{B} uses equations on languages and algebraical approaches to derive several types 
of Dyck language generators for context-free  languages. In all these works, the Dyck language is somehow hidden behind the derivative 
structure of the context-free  language (supplementary brackets are needed to derive a Dyck language generator for a context-free  language). The 
Dyck language provided in this paper is merely found through a \textit{pairwise-renaming procedure} of the nonterminals in 
the original context-free  grammar. Hence, it lies inside the context-free  grammar it describes. Each method used in the literature to prove 
the Chomsky-Sch\"utzenberger theorem provides its own regular language. Our aim is to find a thiner regular language that satisfies the Chomsky-Sch\"utzenberger theorem 
(with respect to the method hereby used) and approaching this language to build a regular superset approximation 
for context-free  languages (likely to be as thiner as possible). 

Note that the concept of a thiner (minimal) regular language, for the Chomsky-Sch\"utzen- berger theorem and for the regular superset approximation is 
\textit{relative}, in the sense that it depends on the structure of the grammar in Dyck normal form used to generate the original context-free  language. 
In \cite{Ce}, \cite{GK1}, \cite{GK2}, and \cite{GK3} it is proved that there is no algorithm that builds, for an arbitrary context-free  
language $L$, the minimal context-free  grammar that generates $L$, where the minimality of a context-free  grammar is considered, in principal, with 
respect to descriptional measures such as number of nonterminals, rules, and loops (i.e., grammatical levels \cite{GK1}, 
encountered during derivations in a context-free  grammar). Consequently, there is no algorithm to build a minimal regular superset 
approximation for an arbitrary context-free  language. Attempts to find optimal regular superset (subsets) approximations for context-free  languages 
can be found in \cite{CS}, \cite{Eg}, \cite{MN}, and \cite{SB}. In Sections 3, 4, and 5 we also illustrate, through several examples, 
the manner in which the regular languages provided by the Chomsky-Sch\"utzenberger theorem and by the regular approximation can be built, 
with regards to the method proposed in this paper. 

\section{Dyck Normal Form}

We assume the reader to be familiar with the basic notions of formal language theory \cite{H}. For an alphabet $X$, 
$X^*$ denotes the free monoid generated by $X$. By $|x|_a$ we denote the number of occurrences of the letter $a$ in 
the string $x\in X^*$, while $|x|$ is the length of $x\in V^*$. We denote by $\lambda$ the empty string. If $X$ is a 
finite set, then $|X|$ is the cardinality of $X$. 
%\vspace{-0.1cm}

\begin{definition}
A context-free grammar\footnote{\textit{A context-free grammar} is denoted by $G=(N, T, P, S)$, 
where $N$ and $T$ are finite sets of \textit{variables} and \textit{terminals}, respectively, $N\cap T= \emptyset$,
$S\in N-T$ is the grammar \textit{axiom}, and $P\subseteq N\times (N\cup T)^*$ is the finite set of \textit{productions}.} 
$G=(N, T, P, S)$ is said to be in {\it Dyck normal form} if it satisfies the following conditions:
\begin{enumerate}
\vspace{-0.1cm}
  \item G is in Chomsky normal form, 
  \vspace{-0.1cm}
  \item if $A\rightarrow a \in P$, $A\in N$, $A\neq S$, $a\in T$, 
   then no other rule in $P$ rewrites $A$, 
   \vspace{-0.1cm}
  \item for each $A\in N$ such that $X\rightarrow AB \in P$
   ($X\rightarrow BA \in P$) there is no other rule in $P$ of
  the form $X'\rightarrow B'A$ ($X'\rightarrow AB'$),
  \vspace{-0.1cm}
  \item for each rules $X\rightarrow AB$, $X'\rightarrow
  A'B$ ($X\rightarrow AB$, $X'\rightarrow AB'$), we have $A=A'$ ($B=B'$).
\end{enumerate}
\end{definition}
%\vspace{-0.1cm}

Note that the reasons for which we introduce the restrictions at items $2 - 4$, are the following. 
The condition at item 2 allows to make a partition between those nonterminals rewritten by nonterminals, and 
those nonterminals rewritten by terminals (with the exception of the axiom). This enables, in Section 2, to define a 
homomorphism from Dyck words to words generated by a grammar in Dyck normal form. Conditions at items 3 and 4 allow to split the 
set of nonterminals into pairwise nonterminals, and thus to introduce bracketed pairs. The next theorem proves that the 
Dyck normal form is correct. 
%\vspace{-0.1cm}

\begin{theorem}
For each context-free  grammar  $G=(N, T, P, S)$ there exists a grammar $G'=(N', T, P', S)$ such that 
$L(G)=L(G')$ where $G'$ is in Dyck normal form.
\end{theorem}
\begin{proof}
Suppose that $G$ is a context-free  grammar in Chomsky normal form. Otherwise, using the algorithm  described in \cite{PL} 
we can convert $G$ into Chomsky normal form. To convert $G$ from Chomsky normal form into Dyck normal form we proceed as follows.\vspace{0.2cm}\\
\textbf{Step 1} We check  whether $P$ contains two (or more) 
rules of the form $A\rightarrow a$, $A\rightarrow b$, $a\neq b$.  If it does, then for each rule 
$A\rightarrow b$, $a\neq b$, a new variable $A_b$ is introduced. We add the new rule $A_b\rightarrow b$, 
and remove $A\rightarrow b$. For each rule $X\rightarrow AB$ ($X\rightarrow BA$) we add the new rule 
$X\rightarrow A_bB$ ($X\rightarrow BA_b$), while for a rule of the form $X\rightarrow AA$ we add three 
new rules $X\rightarrow A_bA$, $X\rightarrow AA_b$, $X\rightarrow A_bA_b$, without removing
the initial rules. We call this procedure an $A_b$-{\it terminal 
substitution} of $A$. For each rule $A\hspace{-0.1cm}\rightarrow a$, $a\in T$, we check 
whether a rule of the form $A\hspace{-0.1cm}\rightarrow B_1B_2$, $B_1,B_2\in N$, exists 
in $P$. If it does, then a new  nonterminal $A_a$ is introduced and an $A_a$-terminal substitution 
of $A$ for the rule $A\rightarrow a$ is performed.\vspace{0.2cm}\\
\textbf{Step 2} Suppose there exist two (or more) rules of the form 
$X\rightarrow AB$ and $X'\rightarrow B'A$. If we have agreed on preserving 
only the left occurrences of $A$ on the right-hand sides, then according to 
condition 3 of Definition 1.1, we have to remove all right occurrences of $A$. 
To do so we introduce a new nonterminal ${}_ZA$ and all right occurrences 
of $A$, preceded at the left side by $Z$, in the right-hand side of a rule, 
are substituted by ${}_ZA$. For each rule that rewrites $A$, $A\rightarrow Y$, 
$Y\hspace{-0.1cm}\in N^2\cup T$, we add a new rule of the form ${}_ZA\hspace{-0.1cm}\rightarrow Y$, preserving 
the rule\\$A\rightarrow Y$. We call this procedure an ${}_ZA$-{\it nonterminal 
substitution} of $A$. According to this procedure, for the rule $X'\rightarrow B'A$, 
we introduce a new nonterminal ${}_{B'}A$, we add the rule $X'\rightarrow B'{}_{B'}A$, 
and remove the rule $X'\rightarrow B'A$. For each rule that rewrites $A$, of the 
form\footnote{This case deals with the possibility of having $Y=B'{}_{B'}A$, 
too.} $A\rightarrow Y$, $Y\in N^2\cup T$, we add a new rule of the form 
${}_{B'}A\rightarrow Y$, preserving the rule $A\rightarrow Y$. \vspace{0.2cm}\\
\textbf{Step 3}  Finally, for each two rules $X\rightarrow AB$, $X'\rightarrow
A'B$ ($X\rightarrow BA$, $X'\rightarrow BA'$) with $A\neq A'$, a new nonterminal 
${}_{A'}B$ ($B_{A'}$) is introduced to replace $B$ from the second rule, and we  
perform an ${}_{A'}B$($B_{A'}$)-nonterminal substitution of $B$, i.e., we add 
$X'\rightarrow A'{}_{A'}B$, and remove $X'\rightarrow A'B$. For each rule that 
rewrites $B$, of the form $B\rightarrow Y$, $Y\in N^2\cup T$, we add a new rule 
${}_{A'}B\rightarrow Y$, preserving $B\rightarrow Y$. In 
the case that $A'$ occurs on the right-hand side of another rule, such
that $A'$ matches at the right side with another nonterminal
different of ${}_{A'}B$, then the procedure described above is repeated
for $A'$, too. 
\vspace*{0.1cm}

Note that, if one of the conditions 2, 3, and 4 in Definition 1.1, has been settled, we 
do not have to resolve it once again in further steps of the procedure. The new grammar 
$G'$ built as described at steps 1, 2, and 3 has the set of nonterminals $N'$ and 
the set of productions $P'$ composed of all nonterminals from $N$ and productions 
from $P$, plus/minus all nonterminals and productions, respectively introduced/removed 
according to the substitutions performed during the above steps. Next we prove that grammars 
$G=(N, T, P, S)$ in Chomsky normal form, and $G'=(N', T, P', S)$ in Dyck normal form, generate the same language. 
Consider the homomorphism  $h_d\hspace{-0.1cm}: N'\cup T \rightarrow N \cup T$ defined by $h_d(x)=x$, 
$x\in T$, $h_d(X)=X$, for $X\in N$, and $h_d(X')=X$ for $X'\in N'-N$, $X\in N$ such that $X'$ is a 
(transitive\footnote{There exist $X_k\in N$, such that $X'$ is an $X'$-substitution of $X_k$, $X_k$ 
is an $X_k$-substitution of $X_{k-1}$,..., and $X_1$ is an $X_1$-substitution of $X$. All of them 
substitute $X$.}) $X'$-substitution of $X$, terminal or not, in the above construction of the grammar $G'$.
\vspace*{0.1cm}

To prove that $L(G') \subseteq L(G)$ we extend $h_d$ to a
homomorphism from $(N'\cup T)^*$ to $(N \cup T)^*$ defined on the
classical concatenation operation. It is straightforward to prove by
induction, that for each $\alpha \Rightarrow^*_{G'} \delta$ we have
$h_d(\alpha) \Rightarrow^*_G h_d(\delta)$. This implies that for any 
derivation of a word $w\in L(G')$, i.e., $S\Rightarrow^*_{G'}w$, we 
have $h_d(S) \Rightarrow^*_G h_d(w)$, i.e., $S \Rightarrow^*_G w$, 
or equivalently, $L(G') \subseteq L(G)$. 
\vspace*{0.1cm}

To prove that $L(G) \subseteq L(G')$ we make use of the CYK
(Cocke-Younger-Kasami) algorithm as described in \cite{PL}.
Let $w=a_1a_2...a_n$ be an arbitrary word in $L(G)$, and $V_{ij}$, $i\leq j$, 
$i,j\in \{1,..., n\}$, be the triangular matrix of size $n\times n$ built with the 
CYK algorithm. Since $w\in L(G)$, we have $S \in V_{1n}$. We prove 
that $w\in L(G')$, i.e., $S \in V'_{1n}$, where $V'_{ij}$, $i\leq j$, $i,j\in \{1,..., n\}$ 
forms the triangular matrix obtained by applying the CYK algorithm to $w$ according to 
$G'$ productions. 
\vspace*{0.1cm}

We consider two relations $\hat h_t \subseteq  (N\cup T)\times (N'\cup T)$ and
$\hat h_{\neg t} \subseteq N \times N'$. The first relation is defined by
$\hat h_t(x)=x$, $x\in T$, $\hat h_t(S)=S$, if $S\rightarrow t$, $t\in T$,
is a rule in $G$, and $\hat h_t(X)=X'$, if $X'$ is a (transitive) $X'$-terminal
substitution\footnote{There may exist several terminal/nonterminal 
substitutions for the same nonterminal $X$. This makes $\hat h_t$/$\hat h_{\neg t}$
to be a relation.} of $X$, and $X\rightarrow t$ is a rule in $G$. Finally, $\hat h_t(X)=X$ if 
$X\rightarrow t \in P$, $t\in T$. The second relation is defined
as $\hat h_{\neg t}(S)=S$, $\hat h_{\neg t}(X)=\{X\}\cup \{ X'| X' \mbox{ is a
(transitive) } X'\mbox{-nonterminal substitution of } X\}$ and $\hat h_{\neg
t}(X)=X$, if there is no substitution of $X$ and no rule of the form $X\rightarrow t$, 
$t\in T$, in $G$. Note that $\hat h_x(X_1\cup X_2)\hspace{-0.1cm}=\hat h_x(X_1)\cup \hat h_x(X_2)$, 
for $X_i \subseteq N$, $i\in \{1,2\}$, $x\in \{t, \neg t\}$. Using $\hat h_t$, 
each rule $X\hspace{-0.2cm}\rightarrow t$ in $P$ has a corresponding set of rules 
$\{X'\hspace{-0.1cm}\rightarrow t| X'\hspace{-0.1cm}\in \hat h_t(X), X\hspace{-0.1cm}\rightarrow t \in P\}$ 
in $P'$. Each rule $A\hspace{-0.1cm}\rightarrow BC$ in $P$ has a corresponding set of rules  
$\{A'\hspace{-0.1cm}\rightarrow B'C'| A'\in \hat h_{\neg t}(A), B'\in \hat h_{\neg t}(B)\cup\\
\hat h_t(B), C'\in \hat h_{\neg t}(C)\cup\hat h_t(C), B' \mbox{ and } C' \mbox{ are pairwise nonterminals, } A\rightarrow BC \in P\}$ in $P'$.
\vspace*{0.1cm}

Consider $V'_{ii}= \hat h_t(V_{ii})$ and $V'_{ij}= \hat
h_{\neg t}(V_{ij})$, $i < j$, $i,j\in \{1,..., n\}$. We claim
that $V'_{ij}$, $i,j\in \{1,..., n\}$, $i\leq j$, defines the triangular matrix obtained by applying 
CYK algorithm to rules that derive $w$ in $G'$. First, observe that for $i=j$, we have 
$V'_{ii}=\hat h_t(V_{ii})=\{A | A\rightarrow a_i\in P'\}$, $i\in \{1,..., n\}$, due to 
the definition of $\hat h_t$. Now let us consider $k=j-i$, $k\in \{1,..., n-1\}$. 
We want to compute $V'_{ij}$, $i<j$.

By definition, we have $V_{ij}=\bigcup_{l=i}^{j-1} \{A| A\rightarrow BC, B\in V_{il}, 
C\in V_{l+1j}\}$, so that $V'_{ij}=\hat h_{\neg t}(V_{ij})\hspace{-0.1cm}=\hat h_{\neg t}(\bigcup_{l=i}^{j-1} \{A|
A\rightarrow BC, B\in V_{il}, C\in V_{l+1j}\})\hspace{-0.1cm}=\bigcup_{l=i}^{j-1}
\hat h_{\neg t}(\{A| A\rightarrow BC, B\in V_{il}, C\in V_{l+1j}\})=
\bigcup_{l=i}^{j-1}\{A'| A'\rightarrow B'C', A'\in \hat h_{\neg t}(A),
B'\in \hat h_{\neg t}(B)\cup \hat h_t(B), B\in V_{il},
C'\in \hat h_{\neg t}(C)\cup\hat h_t(C), C\in V_{l+1j}$,
$B'$ and $C'$ are pairwise nonterminals, $A\rightarrow BC \in P\}$.
Let us explicitly develop the last union.
%\vspace*{0.1cm}

If $k=1$, then $l\in \{i\}$. For each $i\in \{1,..., n-1\}$ we have
$V'_{ii+1}=\{A'| A'\rightarrow B'C', A'\in \hat h_{\neg t}(A), B'\in
\hat h_{\neg t}(B)\cup \hat h_t(B), B\in V_{ii}, C'\in \hat h_{\neg
t}(C)\cup\hat h_t(C), C\in V_{i+1i+1}$, $B'$ and $C'$ are pairwise
nonterminals, $A\rightarrow BC \in P\}$. Due to the fact that $B\in
V_{ii}$ and $C\in V_{i+1i+1}$, $B'$ is a terminal substitution of
$B$, while $C'$ is a terminal substitution of $C$. Therefore, we have
$B'\notin \hat h_{\neg t}(B)$, $C'\notin\hat h_{\neg t}(C)$, so that
$B'\in \hat h_t(B)$, for all $B\in V_{ii}$, and $C'\in \hat h_t(C)$, 
for all $C\in V_{i+1i+1}$, i.e., $B'\in \hat h_t(V_{ii})=V'_{ii}$
and $C'\in \hat h_t(V_{i+1i+1})=V'_{i+1i+1}$. Therefore,
$V'_{ii+1}=\{A'| A'\rightarrow B'C', B'\in V'_{ii}, C'\in
V'_{i+1i+1}\}$.
\vspace*{0.1cm}

If $k\geq 2$, then $l\in \{i, i+1,..., j-1\}$, and
$V'_{ij}=\bigcup_{l=i}^{j-1}\{A'| A'\rightarrow B'C', A'\in \hat h_{\neg t}(A),
B'\in \hat h_{\neg t}(B)\cup \hat h_t(B), B\in V_{il},
C'\in \hat h_{\neg t}(C)\cup\hat h_t(C), C\in V_{l+1j}$,
$B'$ and $C'$ are pairwise nonterminals, $A\rightarrow BC \in P\}$.
We now compute the first set of the above union, i.e., $V'_i=\{A'|
A'\rightarrow B'C', A'\in \hat h_{\neg t}(A), B'\in \hat h_{\neg
t}(B)\cup \hat h_t(B), B\in V_{ii}, C'\in \hat h_{\neg t}(C)\cup\hat
h_t(C), C\in V_{i+1j}$, $B'$ and $C'$ are pairwise nonterminals,
$A\rightarrow BC \in P\}$. By the same reasoning as before, the
condition $B'\in \hat h_{\neg t}(B)\cup \hat h_t(B), B\in V_{ii}$,
is equivalent with $B'\in \hat h_t(V_{ii})=V'_{ii}$.
Because $i+1 \neq j$, $C'$ is a nonterminal substitution of $C$. 
Therefore, $C'\notin \hat h_t(C)$, and the condition
$C'\in \hat h_{\neg t}(C)\cup \hat h_t(C), C\in V_{i+1j}$ is equivalent
with $C'\in \hat h_{\neg t}(V_{i+1j})=V'_{i+1j}$. So that
$V'_i= \{A'| A'\rightarrow B'C', B'\in V'_{ii}, C'\in V'_{i+1j}\}$.
Using the same method for each $l\in \{i+1,..., j-1\}$ we have
 $V'_l=\{A'| A'\rightarrow B'C', A'\in \hat h_{\neg t}(A), B'\in \hat
h_{\neg t}(B)\cup \hat h_t(B), B\in V_{il}, C'\in \hat h_{\neg
t}(C)\cup\hat h_t(C), C\in V_{l+1j}$, $B'$ and $C'$ are pairwise
nonterminals, $A\rightarrow BC \in P\}= \{A'| A'\rightarrow B'C',
B'\in V'_{il}, C'\in V'_{l+1j}\}$. In conclusion, $V'_{ij}=
\bigcup_{l=i}^{j-1} \{A'| A'\rightarrow B'C', B'\in V'_{il}, C'\in
V'_{l+1j}\}$, for each $i,j\in \{1,..., n\}$, i.e., $V'_{ij}$, $i\leq j$, 
contains  the nonterminals of the $n\times n$ triangular matrix computed  
by applying the CYK algorithm to rules that derive $w$ in $G'$.
Because $w\in L(G)$, we have $S\in V_{1n}$. That is equivalent with
$S\in V'_{1n}=\hat h_t(V_{1n})$, if $n=1$, and 
$S\in V'_{1n}=\hat h_{\neg t}(V_{1n})$, if $n>1$, i.e., $w\in L(G')$.
\end{proof}

\begin{corollary}
Let $G$ be a context-free  grammar in Dyck normal form. Any terminal derivation  
in $G$ producing a word of length $n$, $n\geq 1$, takes $2n-1$ steps.
\end{corollary}
%\vspace{-0.1cm}

\begin{proof}
If $G$ is a context-free  grammar in Dyck normal form, then it is also in Chomsky normal form, and all properties of the latter hold.
\end{proof}
%\vspace{-0.1cm}

\begin{corollary}
If $G=(N, T, P, S)$ is a grammar in Chomsky normal form, and
$G'=(N', T, P', S)$ its equivalent in Dyck normal form, then there
exists a homomorphism $h_d\hspace{-0.1cm}:N'\cup T \rightarrow N \cup T$, such
that any derivation tree of $w\in L(G)$ is the homomorphic image of
a derivation tree of the same word in $G'$.
\end{corollary}
%\vspace{-0.1cm}

\begin{proof}
Consider the homomorphism $h_d\hspace{-0.1cm}:N'\cup T \rightarrow N \cup T$  
defined as $h_d(A_t)= h_d({}_ZA)= h_d(A_Z)=A$, for each $A_t$-terminal 
or ${}_ZA$($A_Z$)-nonterminal substitution of $A$, and $h_d(t)=t$, 
$t\in T$. The claim is a direct consequence of the way in which 
the new nonterminals $A_t$, ${}_ZA$, and $A_Z$ have been chosen. 
\end{proof}
\vspace*{0.1cm}

Note that, due to the pairwise-renaming procedure used to reach the Dyck normal form, it may appear that a context-free  grammar in Dyck normal form is more 
ambiguous than the original grammar in Chomsky normal form. However, this is relative. The derivation trees of a certain word have the 
same structure in both grammars, in Chomsky normal form and Dyck normal form (only some ``labels'' of the nodes in these trees differ). The apparent 
ambiguity can be resolved through the homomorphism $h_d$ considered in Corollary 1.4.
\vspace*{0.1cm}

Let $G$ be a grammar in Dyck normal form. To emphasis the pairwise brackets occurring on the 
right-hand side of a rule, and also to make the connection with the Dyck language, 
each pair $(A, B)$, such that there exists a rule of the form 
$X\hspace{-0.1cm}\rightarrow AB$, is replaced by an indexed pair of brackets 
$[_i, ]_i$. In each rule that rewrites $A$ and $B$, we replace $A$ by $[_i$, and $B$ by 
$]_i$, respectively. Next we present an example of the conversion procedure described in 
the proof of Theorem 1.2 along with the homomorphism considered in Corollary 1.4.
\vspace*{-0.4cm}\\

\begin{example}\rm 
Consider the context-free  grammar in Chomsky normal form  $G\hspace{-0.1cm}=\hspace{-0.1cm}(\{E_0,E,E_1,\\
E_2,T,T_1,T_2,R\},\{+,*, a\}, E_0, P')$, where  
$P'=\{E_0\hspace{-0.1cm}\rightarrow a/TT_1/EE_1,E\hspace{-0.1cm}\rightarrow \hspace{-0.1cm}a/TT_1/EE_1, 
T\hspace{-0.1cm}\rightarrow \hspace{-0.1cm}a/TT_1, T_1\hspace{-0.1cm}\rightarrow T_2R, E_1\hspace{-0.1cm}\rightarrow \hspace{-0.1cm}E_2T, T_2\hspace{-0.1cm}\rightarrow \hspace{-0.1cm}*, E_2\hspace{-0.1cm}\rightarrow \hspace{-0.1cm}+, R\hspace{-0.1cm}\rightarrow a\}$.
\vspace{0.1cm}

To convert $G$ into Dyck normal form, with respect to Definition 1.1, item 2, we first remove 
$E\rightarrow a$ and $T\rightarrow a$. Then, according to
item 3, we remove the right occurrence of $T$ from the rule
$E_1\rightarrow E_2T$, along with other transformations that may be
required after completing these procedures. Let $E_3$ and $T_3$ be two 
new nonterminals. We remove  $E\rightarrow a$ and $T\rightarrow a$, and 
add the rules $E_3\rightarrow a$, $T_3\rightarrow a$, $E_0\rightarrow E_3E_1$, 
$E_0\rightarrow T_3T_1$, $E\rightarrow E_3E_1$, $E\rightarrow T_3T_1$,
$E_1\rightarrow E_2T_3$, $T\rightarrow T_3T_1$.
Let $T'$ be the new nonterminal that replaces the right occurrence
of $T$. We add the rules $E_1\rightarrow E_2T'$, $T'\rightarrow TT_1$,
$T'\rightarrow T_3T_1$, and remove $E_1\rightarrow E_2T$.
We repeat the procedure with $T_3$ (added in the previous step), i.e.,
we introduce a new nonterminal $T_4$, remove $E_1\rightarrow E_2T_3$,
add $E_1\rightarrow E_2T_4$ and  $T_4 \rightarrow a$. 

Due to the new nonterminals $E_3$, $T_3$, $T_4$, item 4  does not hold. 
To have accomplished this condition, we introduce three new nonterminals 
$E_4$ to replace $E_2$ in $E_1\rightarrow E_2T_4$, $E_5$ to replace $E_1$ 
in $E_0\rightarrow E_3E_1$ and $E\rightarrow E_3E_1$, and $T_5$ to
replace $T_1$ in  $E_0\rightarrow T_3T_1$ and $E\rightarrow T_3T_1$.
We remove all the above rules and add the new rules $E_1\rightarrow
E_4T_4$, $E_4 \rightarrow +$, $E_0\rightarrow E_3E_5$, $E\rightarrow
E_3E_5$, $E_5\rightarrow E_2T'$, $E_5\rightarrow E_4T_4$,
$E_0\rightarrow T_3T_5$, $E\rightarrow T_3T_5$, and $T_5\rightarrow
T_2R$. 
\vspace*{0.1cm}

The Dyck normal form of $G'$, in bracketed notation, is 
$G''=(\{E_0,[_1,[_2,...,[_7,]_1,]_2, ...,]_7\},\\ 
\{+,*, a\}, E_0, P'')$, $P''\hspace{-0.1cm}=\hspace{-0.1cm}\{E_0\hspace{-0.1cm}\rightarrow a/[_1\hspace{0.1cm}]_1/[_2\hspace{0.1cm}]_2/[_3\hspace{0.1cm}]_3/[_4\hspace{0.1cm}]_4, 
[_1\rightarrow[_1\hspace{0.1cm}]_1/[_4\hspace{0.1cm}]_4,[_2\rightarrow [_1\hspace{0.1cm}]_1/[_2\hspace{0.1cm}]_2/[_3\hspace{0.1cm}]_3/[_4\hspace{0.1cm}]_4,\\
]_1\rightarrow [_7\hspace{0.1cm}]_7,\hspace{0.1cm}
]_2\rightarrow[_5\hspace{0.1cm}]_5/[_6\hspace{0.1cm}]_6,\hspace{0.1cm}
]_3\rightarrow[_5\hspace{0.1cm}]_5/[_6\hspace{0.1cm}]_6, \hspace{0.1cm}
 ]_4\rightarrow[_7\hspace{0.1cm}]_7, \hspace*{0.1cm}]_5\rightarrow[_1\hspace{0.1cm}]_1/[_4\hspace{0.1cm}]_4, 
[_3\hspace{0.1cm}\rightarrow a,\hspace{0.1cm}[_4\hspace{0.1cm}\rightarrow a,[_5\hspace{0.1cm}\rightarrow +,$\\$[_6\hspace{0.1cm}\rightarrow +,\hspace{0.1cm}]_6\rightarrow a, 
[_7\hspace{0.1cm}\rightarrow *, \hspace{0.1cm}]_7\rightarrow a\}$, 
where $([_T, ]_{T_1})= ([_1,]_1)$, $([_E, ]_{E_1})= ([_2,]_2)$, $([_{E_3}, ]_{E_5})= ([_3,]_3)$, 
$([_{T_3}, ]_{T_5})= ([_4,]_4)$, $([_{E_2}, ]_{T'})= ([_5,]_5)$, $([_{E_4}, ]_{T_4})= ([_6,]_6)$,  
$([_{T_2}, ]_{R})= ([_7,]_7)$. 
\vspace*{0.1cm}

The homomorphism $h_{d}$ is defined as $h_d\hspace{-0.1cm}:N'\cup T \rightarrow N'' \cup T$,  
$h_d(E_0)= E_0$, $h_d([_2)= h_d([_3)=E$, $h_d(]_2)= h_d(]_3)=E_1$,  
$h_d([_5)= h_d([_6)=E_2$, $h_d([_1)= h_d(]_5)=h_d([_4)=h_d(]_6)=T$,
$h_d(]_1)= h_d(]_4)=T_1$, $h_d([_7)=T_2$, $h_d(]_7)=R$, $h_d(t)= t$, for each $t\in T$. 
\vspace*{0.1cm}

The string $w=a*a*a+a$ is a word in $L(G'')=L(G)$ generated, for instance, 
by a leftmost derivation $D$ in $G''$ as follows.
\vspace*{0.1cm}

$D\hspace{-0.1cm}: E_0\Rightarrow [_2\hspace{0.1cm}]_2\Rightarrow
[_1\hspace{0.1cm}]_1\hspace{0.1cm}]_2\Rightarrow
[_4\hspace{0.1cm}]_4\hspace{0.1cm}]_1\hspace{0.1cm}
]_2\Rightarrow a\hspace{0.1cm}]_4\hspace{0.1cm}]_1\hspace{0.1cm}
]_2\Rightarrow a\hspace{0.1cm}[_7\hspace{0.1cm}]_7\hspace{0.1cm}]_1\hspace{0.1cm}
]_2\Rightarrow a\hspace{0.1cm}*\hspace{0.1cm}]_7\hspace{0.1cm}]_1\hspace{0.1cm}
]_2\Rightarrow a\hspace{0.1cm}*\hspace{0.1cm}a\hspace{0.1cm}]_1\hspace{0.1cm}
]_2\Rightarrow a\hspace{0.1cm}*\hspace{0.1cm}a\hspace{0.1cm}[_7\hspace{0.1cm}]_7
\hspace{0.1cm} ]_2\Rightarrow a\hspace{0.1cm}*\hspace{0.1cm}a\hspace{0.1cm}*\hspace{0.1cm}]_7
\hspace{0.1cm} ]_2\Rightarrow a\hspace{0.1cm}*\hspace{0.1cm}a\hspace{0.1cm}*\hspace{0.1cm}a
\hspace{0.1cm} ]_2\Rightarrow a\hspace{0.1cm}*\hspace{0.1cm}a\hspace{0.1cm}*\hspace{0.1cm}a
\hspace{0.1cm} [_6\hspace{0.1cm}]_6\Rightarrow
a\hspace{0.1cm}*\hspace{0.1cm}a\hspace{0.1cm}*\hspace{0.1cm}a
\hspace{0.1cm} +\hspace{0.1cm}]_6\Rightarrow 
a\hspace{0.1cm}*\hspace{0.1cm}a\hspace{0.1cm}*\hspace{0.1cm}a
\hspace{0.1cm} +\hspace{0.1cm}a$. 

Applying $h_{d}$ to $D$, in $G''$, we obtain a derivation of $w$ in $G'$. If we consider $\cal T$ 
the derivation tree of $w$ in $G$, and $\cal T'$ the derivation tree of $w$ in $G''$, then $\cal T$ 
is the homomorphic image of $\cal T'$ through $h_d$. 
\end{example}
%\vspace*{-0.1cm}

\section{Characterizations of Context-Free  Languages by Dyck\\Languages}
%\vspace*{-0.1cm}

\begin{definition}\rm
Let $G_k=(N_k, T, P_k, S)$ be a context-free  grammar in Dyck normal form with $|N_k-\{S\}|=2k$. 
Let $D\hspace{-0.1cm}:S\Rightarrow u_1 \Rightarrow u_2\Rightarrow... \Rightarrow u_{2n-1}=w$,
$n\geq 2$, be a leftmost derivation of $w\in L(G)$. The {\it trace-word} of $w$ 
associated with the derivation $D$, denoted as $t_{w,D}$, is defined as the 
concatenation of nonterminals consecutively rewritten in $D$, excluding the axiom. 
The {\it trace-language} associated with $G_k$, denoted by $\L(G_k)$, is  
$\L(G_k)=\{t_{w,D}|\mbox{ for any } w\in L(G_k), \mbox{ and any leftmost derivation } D \mbox{ of } w \}$.
\end{definition}
%\vspace*{-0.1cm}

Note that $t_{w,D}$, $w\in L(G)$, can also be read from the derivation tree in the 
depth-first search order starting with the root, but ignoring the root and the leaves.
The trace-word associated with $w$ and the leftmost derivation $D$ in 
Example 2.5 is $t_{a*a*a+a,D}=[_E\hspace{0.1cm}[_T\hspace{0.1cm}[_{T_3}
\hspace{0.1cm}]_{T_5}\hspace{0.1cm}[_{T_2}\hspace{0.1cm}]_R
\hspace{0.1cm}]_{T_1}\hspace{0.1cm}[_{T_2}\hspace{0.1cm}]_R
\hspace{0.1cm}]_{E_1}\hspace{0.1cm}[_{E_4}\hspace{0.1cm}]_{T_4}$.

\begin{definition}\rm
A {\it one-sided Dyck language} over $k$ letters, $k\geq 1$, is a context-free  language defined by the 
grammar $\Gamma_k=(\{S\}, T_k, P, S)$, where $T_k=\{[_1,\hspace{0.1cm}[_2,...,[_k, \hspace{0.1cm}]_1,\hspace{0.1cm}]_2,
...,\hspace{0.1cm}]_k \}$ and $P=\{S\rightarrow [_i\hspace{0.1cm}S\hspace{0.1cm}]_i, 
S\rightarrow SS,S\rightarrow [_i\hspace{0.1cm}]_i \hspace{0.1cm}|\hspace{0.1cm}
1\leq i \leq k\}.$ 
\end{definition}
%\vspace*{-0.1cm}

Let $G_k=(N_k, T, P_k, S)$ be a context-free  grammar in Dyck normal form. To emphasize 
possible relations between the structure of trace-words in $\L(G_k)$ and 
the structure of words in the Dyck language, and also to keep control of 
each bracketed pair occurring on the right-hand side of each rule in $G_k$, 
we fix $N_k=\{S, [_1,\hspace{0.1cm}[_2,...,[_k,\hspace{0.1cm}]_1,\hspace{0.1cm}]_2,...,\hspace{0.1cm}]_k \}$, 
and $P_k$ to be composed of rules of the forms $X \rightarrow
[_i\hspace{0.1cm}]_i$, $1\leq i \leq k$, and $Y \rightarrow t$, $X, Y
\in N_k$,  $t \in T$. From \cite{IJR} we have adopted the next characterizations 
of $D_k$, $k\geq 1$, (Definition 2.3, and Lemmas 2.4 and 2.5). 
%\vspace*{-0.1cm}

\begin{definition}\rm
For a string  $w$, let $w_{i:j}$ be its substring starting at the
$i^{th}$ position and ending at the $j^{th}$ position. Let $h$ be a
homomorphism defined as follows: 

\hspace{1.5cm} $h([_1)=h([_2)=...=h([_k)=[_1,$ \hspace{0.7cm}
$h(]_1)=h(]_2)=...=h(]_k)=]_1$.\\
Let $w\in D_k$, $1 \leq i\leq j \leq |w|$, where $|w|$ is the length
of $w$. We say that ($i$, $j$) is a {\it matched pair} of $w$, if $h(w_{i:j})$
is {\it balanced}, i.e., $h(w_{i:j})$ has an equal number of $[_1$'s and
$]_1$'s and, in any prefix of $h(w_{i:j})$, the number of $[_1$'s is
greater than or equal to the number of $]_1$'s.
\end{definition}
%\vspace*{-0.1cm}

\begin{lemma}
A string $w\in \{[_1 \hspace{0.1cm},]_1\}^*$ is in $D_1$ if and only
if it is balanced. 
\end{lemma}
%\vspace*{-0.1cm}

Consider the homomorphisms defined as follows (where $\lambda$ is the empty string)

$h_1([_1)=[_1,$ \hspace{0.1cm} $h_1(]_1)=]_1,$ \hspace{0.1cm}
$h_1([_2)=h_1(]_2)=...=h_1([_k)=h_1(]_k)=\lambda,$

$h_2([_2)=[_1,$ \hspace{0.1cm} $h_2(]_2)=]_1,$ \hspace{0.1cm}
$h_2([_1)=h_2(]_1)=...=h_2([_k)=h_2(]_k)=\lambda,$.\hspace{0.1cm}.\hspace{0.1cm}.\hspace{0.1cm}.\hspace{0.1cm}.

%.\hspace{0.1cm}.\hspace{0.1cm}.\hspace{0.1cm}.\hspace{0.1cm}.\hspace{0.1cm}.\hspace{0.1cm}.\hspace{0.1cm}.
%\hspace{0.1cm}.\hspace{0.1cm}.\hspace{0.1cm}.\hspace{0.1cm}.\hspace{0.1cm}.\hspace{0.1cm}.\hspace{0.1cm}.\hspace{0.1cm}%.\hspace{0.1cm}

$h_k([_k)=[_1,$ \hspace{0.1cm} $h_k(]_k)=]_1,$ \hspace{0.1cm}
$h_k([_1)=h_k(]_1)=...=h_k([_{k-1})=h_k(]_{k-1})=\lambda.$
%\vspace*{-0.1cm}

\begin{lemma}
We have $w \in D_k$, $k\geq 2$, if and only if the following
conditions hold: $i$) $(1$, $|w|)$ is a matched pair, and $ii$) for all 
matched pairs $(i$, $j)$, $h_k(w_{i:j})$  are in $D_1$, where $k\geq 1$.
\end{lemma}
%\vspace*{-0.4cm}

\begin{definition}\rm
Let $w\in D_k$, ($i$, $j$) is a {\it nested pair} of $w$ if ($i$, $j$) is
a matched pair, and either $j=i+1$, or ($i+1$, $j-1$) is a matched pair.
\end{definition}
%\vspace*{-0.4cm}

\begin{definition}\rm
 Let $w\in D_k$ and ($i$, $j$) be a matched pair of $w$. We say
 that ($i$, $j$) is {\it reduci-\\ble} if there exists an integer 
 $j'$, $i < j' \hspace*{-0.1cm}< j$, such that ($i$, $j'$) and ($j'+1$, $j$)
 are matched pairs.\\
\end{definition}
\vspace*{-0.5cm}

Let $w\in D_k$, if ($i$, $j$) is a nested pair of $w$ then ($i$, $j$) 
is an irreducible pair. If ($i$, $j$) is a nested pair of $w$ then 
($i+1$, $j-1$) may be a reducible pair.
%\vspace*{-0.1cm}*

\begin{theorem} The trace-language associated with a context-free  grammar,
$G=(N_k, T, P_k, S)$ in Dyck normal form, with $|N_k|=2k+1$, is a 
subset of $D_k$. 
\end{theorem}
\vspace*{-0.1cm}

\begin{proof}
Let $N_k=\{S,[_1,...,[_k,\hspace{0.1cm}]_1,...,\hspace{0.1cm}]_k\}$ be the set of 
nonterminals, $w\in L(G)$, and $D$ a leftmost derivation of $w$. We show that any 
subtree of the derivation tree, read in the depth-first search order, by  
ignoring the root and the terminal nodes, corresponds to a matched pair in  
$t_{w,D}$. In particular, $(1, |t_{w,D}|)$ will be a matched pair. Denote by 
${t_{w,D}}_{i:j}$ the substring of $t_{w,D}$ starting at the $i^{th}$ position 
and ending at the $j^{th}$ position of $t_{w,D}$. We show that for all matched 
pairs $(i, j)$, $h_{k'}({t_{w,D}}_{i:j})$ belong to $D_1$, $1\leq k'\leq k$. We 
prove these claims by induction on the height of subtrees.
\vspace{0.1cm}

{\it Basis.} Certainly, any subtree of height $n=1$, read in the
depth-first search order, looks like $[_i \hspace{0.1cm}]_i$, $1\leq i\leq
k$. Therefore, it satisfies the above conditions.
\vspace{0.1cm}

{\it Induction step.} Assume that the claim is true for all subtrees of 
height $\hbar$, $\hbar<n$, and we prove it for $\hbar=n$. Each subtree of height $n$ 
can have one of the following structures. The level $0$ of the subtree is marked by 
a left or right bracket. This bracket will not be considered when we read the subtree. 
Denote by $[_{m}$ the left son of the root. Then the right son is
labeled by $]_{m}$. They are the roots of a left and right subtree,
for which at least one has the height $n-1$.

Suppose that both subtrees have the height $1\leq \hbar \leq n-1$. By the induction 
hypothesis, let us further suppose that the left subtree corresponds to the matched 
pair $(i_l,j_l)$, and the right subtree corresponds to the matched pair $(i_r, j_r)$, 
$i_r=j_l+2$, because the position $j_l+1$ is taken by  $]_{m}$. 
As $h$ is a homomorphism, we have $h({t_{w,D}}_{i_l-1:j_r})=
h([_m{t_{w,D}}_{i_l:j_l}]_m{t_{w,D}}_{j_l+2:j_r})=h([_m)h({t_{w,D}}_{i_l:j_l})
h(]_m)h({t_{w,D}}_{j_l+2:j_r})$. Therefore, $h({t_{w,D}}_{i_l-1:j_r})$ satisfies 
all conditions in Definition 2.3, and thus $(i_l-1,j_r)$ that corresponds to the 
considered subtree of height $n$, is a matched pair. By the induction hypothesis, 
$h_{k'}({t_{w,D}}_{i_l:j_l})$ and $h_{k'}({t_{w,D}}_{i_r:j_r})$ are in $D_1$, $1\leq k'\hspace{-0.1cm}\leq k$. 
Hence, $h_{k'}({t_{w,D}}_{i_l-1:j_r})\hspace{-0.1cm}=\hspace{-0.1cm}h_{k'}([_m)h_{k'}({t_{w,D}}_{i_l:j_l})h_{k'}(]_m)
h_{k'}({t_{w,D}}_{j_l+2:j_r})\in\{h_{k'}({t_{w,D}}_{i_l:j_l})h_{k'}({t_{w,D}}_{j_l+2:j_r}),
[_1h_{k'}({t_{w,D}}_{i_l:j_l})]_1h_{k'}({t_{w,D}}_{j_l+2:j_r})\}$
belong to $D_1$, $1\leq k'\leq k$. Note that in this case the matched pair $(i_l-1,j_r)$ is 
reducible into $(i_l-1,j_l+1)$ and $(j_l+2,j_r)$, where $(i_l-1,j_l+1)$ corresponds to the 
substring ${t_{w,D}}_{i_l-1:j_l+1}=[_m{t_{w,D}}_{i_l:j_l}]_m$. We refer to 
this structure as the {\it left embedded subtree}, i.e., $(i_l-1,j_l+1)$ 
is a nested pair. A similar reasoning is applied for the case when one 
of the subtrees has the height $0$. Analogously, it can be shown that 
the initial tree corresponds to the matched pair $(1, |t_{w,D}|)$, i.e., 
the first condition of Lemma 2.5 holds.
So far, we have proved that each subtree of the derivation tree, and
also each left embedded subtree, corresponds to a matched pair
$(i,j)$ and $(i_l,j_l)$, such that $h_{k'}({t_{w,D}}_{i:j})$ and
$h_{k'}([_m{t_{w,D}}_{i_l:j_l}]_m)$, $1\leq k'\leq k$, are in $D_1$.
\vspace{0.1cm}

Next we show that all matched pairs from $t_{w,D}$ correspond only
to subtrees, or left embedded subtrees, from the derivation tree. To derive 
a contradiction, let us suppose that there exists a matched pair
$(i, j)$ in $t_{w,D}$, that does not correspond to any subtree, or
left embedded subtree, of the derivation tree read in the depth-first search 
order. We show that this leads to a contradiction.
\vspace{0.1cm}

Since $(i, j)$ does not correspond to any subtree, or left embedded subtree, there 
exist two adjacent subtrees $\theta_1$ (a left embedded subtree) and  $\theta_2$ 
(a right subtree) such that $(i, j)$ is composed of two adjacent ``subparts'' of $\theta_1$ 
and $\theta_2$. In terms of matched pairs, if $\theta_1$ corresponds to the matched pair 
$(i_1, j_1)$ and $\theta_2$ corresponds to the matched pair $(i_2, j_2)$, such that $i_2=j_1+2$, 
then there exists a suffix $s_{i_1-1:j_1+1}$ of ${t_{w,D}}_{i_1-1:j_1+1}$, and a prefix 
$p_{i_2:j_2}$ of ${t_{w,D}}_{i_2:j_2}$, such that ${t_{w,D}}_{i:j}=s_{i_1-1:j_1+1}p_{i_2:j_2}$. 
Furthermore, without loss of generality, we assume that $(i_1, j_1)$ and $(i_2, j_2)$ are nested 
pairs. Otherwise, the matched pair $(i, j)$ can be ``narrowed'' until $\theta_1$ and $\theta_2$
are characterized by two nested pairs. If $(i_1, j_1)$ is a nested pair, then so is $(i_1-1, j_1+1)$. 
As $s_{i_1-1:j_1+1}$ is a suffix of ${t_{w,D}}_{i_1-1:j_1+1}$ and $(i_1-1, j_1+1)$ is a matched 
pair, with respect to Definition 2.3, the number of $]_1$'s in $h(s_{i_1-1:j_1+1})$ is greater than 
or equal to the number of $[_1$'s in $h(s_{i_1-1:j_1+1})$. On the other hand,  
$s_{i_1-1:j_1+1}$ is also a prefix of ${t_{w,D}}_{i:j}$, because $(i, j)$ 
is a matched pair, by the induction hypothesis. Therefore, the number of $[_1$'s in 
$h(s_{i_1-1:j_1+1})$ is greater than or equal to the number of $]_1$'s in $h(s_{i_1-1:j_1+1})$. 
Hence, the only possibility for $s_{i_1-1:j_1+1}$ to be, in the same time, a suffix for ${t_{w,D}}_{i_1-1:j_1+1}$ 
and a prefix for ${t_{w,D}}_{i:j}$, is the equality between the number of $[_1$'s and $]_1$'s 
in $h(s_{i_1-1:j_1+1})$. This property holds if and only if $s_{i_1-1:j_1+1}$
corresponds to a matched pair in ${t_{w,D}}_{i_1-1:j_1+1}$, i.e., if $i_s$ and $j_s$ are
the start and the end positions of $s_{i_1-1:j_1+1}$ in ${t_{w,D}}_{i_1-1:j_1+1}$, then
$(i_s, j_s)$ is a matched pair. Thus,  $(i_1-1, j_1+1)$ is a reducible pair
into $(i_1-1, i_s-1)$ and $(i_s, j_s)$, where $j_s=j_1+1$. We have reached a 
contradiction, i.e., $(i_1-1, j_1+1)$ is reducible. 
\vspace{0.1cm}

Therefore, the matched pairs in $t_{w,D}$ correspond to subtrees,
or left embedded subtrees, in the derivation tree. For these matched
pairs we have already proved that they satisfy Lemma 2.5. Accordingly,
$t_{w,D} \in D_k$, and consequently the trace-language associated
with $G$ is a subset of $D_k$. 
\end{proof}

\begin{theorem}
Given a context-free  grammar $G$ there exist an integer $K$, a homomorphism $\varphi$, and 
a subset $D'_K$ of the Dyck language $D_K$, such that $L(G)=\varphi(D'_K)$.  
\end{theorem}

\begin{proof}
Let $G$ be a context-free  grammar and $G_k=(N_k, T, P_k, S)$ be the Dyck normal form of
$G$, such that $N_k=\{S, [_1,...,[_k,\hspace{0.1cm}]_1,...,]_k \}$. Let 
$\L(G_k)$ be the trace-language associated with $G_k$. Consider $\{t_{k+1},...,t_{k+p}\}$ the 
ordered subset of $T$, such that $S\rightarrow t_{k+i} \in P$, $1\leq i\leq p$.  
We define $N_{k+p}=N_k \cup \{[_{t_{k+1}}, ..., [_{t_{k+p}}, ]_{t_{k+1}}, ... 
]_{t_{k+p}}\}$, and $P_{k+p}=P_k \cup\{ S\rightarrow [_{t_{k+i}}]_{t_{k+i}}, 
[_{t_{k+i}} \rightarrow t_{k+i}, ]_{t_{k+i}} \rightarrow  \lambda|  S\rightarrow t_{k+i} 
\in P, 1\leq i \leq p\}$. The new grammar $G_{k+p}=(N_{k+p}, T, 
P_{k+p}, S)$ generates the same language as $G_k$. 
\vspace{0.1cm}

Let $\varphi\hspace{-0.1cm}:(N_{k+p}-\{S\})^*\rightarrow T^*$ be the homomorphism
defined by $\varphi(N)=\lambda$, for each rule of the form
$N\rightarrow XY$, $N, X, Y\in N_k-\{S\}$, and $\varphi(N)=t$, for
each rule of the form $N\rightarrow t$, $N\in N_k-\{S\}$, and $t\in
T$, $\varphi([_{k+i})=t_{k+i}$, and $\varphi(]_{k+i})=\lambda$, for
each $1\leq i \leq p$. Obviously, $L=\varphi(D'_K)$, where
$K=k+p$, $D'_K=\L(G_k)\cup{L_p}$, and $L_p=\{[_{t_{k+1}} \hspace{0.1cm}]_{t_{k+1}}, ..., 
[_{t_{k+p}}\hspace{0.1cm}]_{t_{k+p}}\}$.
\end{proof}
\vspace{0.2cm}

In the sequel, grammar $G_{k+p}$  is called the {\it extended grammar} of $G_{k}$. $G_{k}$ has an 
extended grammar if and only if $G_k$ (or $G$) has rules of the form $S\rightarrow t$, $t\in T\cup\{\lambda\}$.
If $G_k$ does not have an extended grammar then $D'_K=D'_k=\L(G_k)$.

\section{On the Chomsky-Sch$\rm\ddot{\textbf{u}}$tzenberger Theorem}
%\vspace{-0.1cm}

Let $G_k=(N_k, T, P_k, S)$ be an arbitrary context-free  grammar in Dyck normal form, with $N_k=\{S, [_1,...,[_k,\hspace{0.1cm}]_1,...,]_k \}$. 
and $\varphi\hspace{-0.1cm}: (N_k-\{S\})^* \rightarrow T^*$ the restriction of the homomorphism $\varphi$ in the proof of Theorem 2.9.  
%defined as $\varphi(N)=\lambda$, for each rule of the form $N\hspace{-0.1cm}\rightarrow \hspace{-0.1cm}XY$, $N, 
%X, Y\hspace{-0.1cm}\in N_k$, $N\neq S$, and $\varphi(N)=t$, for each rule of the form $N\rightarrow t$, $N\in N_k$, $N\neq S$, $t\in T$. 
We divide $N_k$ into three main sets $N^{(1)}$, $N^{(2)}$, $N^{(3)}$ as follows: \\
\hspace*{0.5cm} 1. $[_i$ and $]_i$ belong to $N^{(1)}$ if and only if $\varphi ([_i)=t$ and $\varphi (]_i)=t'$, $t, t' \in T$,\\ 
\hspace*{0.5cm} 2. $[_i$ and $]_i$ belong to $N^{(2)}$ if and only if $\varphi ([_i)=t$ and $\varphi(]_i)=\lambda$, or vice versa $\varphi ([_i)=\lambda$ \hspace*{1cm} and $\varphi (]_i)=t$, $t\in T$,\\ 
\hspace*{0.5cm} 3. $[_i$, $]_i\in N^{(3)}$ if and only $\varphi ([_i)=\lambda$ and $\varphi (]_i)=\lambda$. 
\vspace{0.1cm}

Certainly, $N_k-\{S\}=N^{(1)} \cup N^{(2)}\cup N^{(3)}$ and $N^{(1)}\cap N^{(2)}\cap N^{(3)}=\emptyset$. 
$N^{(2)}$ is further divided into $N^{(2)}_l$ and $N^{(2)}_r$, where $N^{(2)}_l$ contains those pairs 
$[_i, ]_i\in N^{(2)}$ such that $\varphi ([_i)\neq\lambda$, while $N^{(2)}_r$ contains those pairs 
$[_i, ]_i \in N^{(2)}$ such that $\varphi (]_i)\neq\lambda$. 
\vspace*{-0.1cm}

\begin{definition}\rm
A grammar $G_k$ is in linear-Dyck normal form if $G_k$ is in Dyck normal form and  $N^{(3)} = \emptyset$.  
\end{definition}
\vspace*{-0.3cm}
\begin{theorem}
For each linear grammar $G$, there exits a grammar $G_k$ in linear-Dyck normal form such that $L(G) = L(G_k)$, and vice versa. 
\end{theorem}
\vspace*{-0.1cm}
\begin{proof}
Each linear grammar $G$, in standard form, is composed of rules of the forms 
$X\rightarrow \lambda$, $X\rightarrow t$, $X\rightarrow t_1Y$, 
$X\rightarrow Yt_2$, $X\rightarrow t_1Yt_2$, $t, t_1, t_2\in T$, $X$, $Y \in N$. 
Transforming $G$ into Chomsky normal form, and then into the Dyck normal form, we obtain a 
grammar $G_k$ in linear-Dyck normal form. Since the standard form for linear languages, Chomsky normal form, and 
Dyck normal form are weakly equivalent we obtain $L(G) = L(G_k)$. The converse statement is trivial.  
\end{proof} 

Next we consider more closely the structures of the derivation trees associated with words 
generated by linear and context-free  grammars in linear-Dyck normal form and Dyck normal form, respectively. We are interested on  
the structure of the trace-words associated with words generated by these grammars. 

Let $G_k=(N_k, T, P_k, S)$ be an arbitrary (linear) context-free  grammar in (linear-)Dyck normal form, and $L(G_k)$ the 
language generated by this grammar. Let $w\in L(G_k)$, $D$ a leftmost derivation of $w$, and 
$t_{w,D}$ the trace-word of $w$ associated with $D$. From the structure of the derivation tree, read 
in the depth-first search order, it is easy to observe that each bracket $[_i$, such that 
$[_i, ]_i\in N^{(1)}$, is immediately followed, in $t_{w,D}$ by its pairwise $]_i$. The same 
property holds for those pairs $[_i, ]_i \in N^{(2)}_l$. If $[_i, ]_i \in N^{(2)}_r\cup N^{(3)}$ 
then the pair $[_i, ]_i$ should embed a left subtree, i.e., the case of the left embedded subtree in 
the proof of Theorem 2.8. In this case the bracket $[_i$ may have a left, long distance, placement from 
its pairwise $]_i$, in $t_{w,D}$. 

Suppose that $G_k$ is a linear grammar in linear-Dyck normal form, i.e., $N^{(3)}=\emptyset$, such that 
$N^{(2)}_l\neq \emptyset$ and $N^{(2)}_r\neq \emptyset$. Each word $w=a_1a_2...a_n\in L(G_k)$, of an 
arbitrary length $n$, has the property that there exists an index $n_t$, $1\leq n_t \leq n-1$, 
and a unique pair\footnote{To emphasize which of the brackets in the pair $([_i, ]_i)$ produces a terminal, 
we also use the notation $[_i, ]^t_i$ if and only if $[_i, ]_i\in N^{(2)}_r$, $[^t_i, ]_i$ if and only 
if $[_i, ]_i \in N^{(2)}_l$, and $[^t_i, ]^t_i$ if and only if $[_i, ]_i\in N^{(1)}$.}  
$[^t_j,]^t_j \in N^{(1)}$,  such that $[^t_j\rightarrow \hspace{-0.1cm}a_{n_t}$ and 
$]^t_j\hspace{-0.1cm}\rightarrow \hspace{-0.1cm}a_{n_t+1}$. Using the homomorphism $\varphi$ in Theorem 2.9, 
we  have $\varphi([^t_j)=a_{n_t}$ and  $\varphi(]^t_j)=a_{n_t+1}$. For the position $n_t$ already 
``marked'', there is no other position in $w$ with the above property. We call $[^t_j \hspace{0.1cm}]^t_j$ 
the {\it core segment} of the trace-word $t_{w,D}$. Trace-words of words generated by context-free  grammars in Dyck normal form have more 
than one core segment. Each core segment induces in a trace-word (both for linear and context-free  languages) a symmetrical 
distribution of right brackets in $N^{(2)}_r\cup N^{(3)}$ (always placed at the right side of the core segment) 
according to left brackets in $N^{(2)}_r\cup N^{(3)}$ (always placed at the left side of the respective core). 
The structure of the trace-word of a word $w\in L(G_k)$, for a grammar $G_k$ in linear-Dyck normal form, is depicted in 
(\ref{p21}), where by vertical lines we emphasize the image through the homomorphism $\varphi$ of each bracket occurring 
in $t_{w,D}$. 
%Next we refer to the prefix of $t_{w,D}$,  that ends in a pair of brackets in $N^{(1)}$, 
%as the ``first half'' of $t_{w,D}$, and to the suffix of $t_{w,D}$ that begins in a pair of brackets in $N^{(1)}$, 
%as the ``second half'' of $t_{w,D}$.  

\vspace{0.2cm}

\centerline{\hspace{0.5cm}$t_{w,D}$ = \noindent \begin{tabular}{lllllllllll}
\hspace{-0.2cm}$\textcolor[rgb]{1.00,0.50,0.00}{[_{i_1}\hspace{0.2cm}...}$&$\textcolor[rgb]{1.00,0.50,0.00}{[_{i_{k_1}}}$&$\hspace{-0.2cm}\textcolor[rgb]{0.60,0.20,0.00}{{\large{[}_{j_1}}}$&$\textcolor[rgb]{0.60,0.20,0.00}{\large{]}_{j_1}}$&$\hspace{-0.1cm}\textcolor[rgb]{0.20,0.10,0.90}{[_{i_{k_1+1} ...}}$&$\textcolor[rgb]{0.20,0.10,0.90}{\hspace{-0.1cm}[_{i_{k_2}}}$&$\hspace{-0.2cm}\textcolor[rgb]{0.60,0.20,0.00}{\large{[}_{j_2}}$&$\textcolor[rgb]{0.60,0.20,0.00}{\large{]}_{j_2}}...$&$\hspace{-0.1cm}\textcolor[rgb]{0.60,0.20,0.00}{\large{[}_{j_{n_t-1}}}$&$
\textcolor[rgb]{0.60,0.20,0.00}{\large{]}_{j_{n_t-1}}}$&$\hspace{-0.1cm}\textcolor[rgb]{0.30,0.80,0.50}{[_{i_{k_{n_t-1}+1}}...\hspace{0.1cm}[_{i_{k_{n_t}}}}$\\
\hspace{-0.2cm}$|\hspace{0.5cm}...$&$|$&$\hspace{-0.1cm}|$&$|$&$|\hspace{0.5cm}...$&$\hspace{-0.1cm}|$&$\hspace{-0.1cm}|$&
$|\hspace{0.2cm}...$&$\hspace{0.2cm}|$&$|$&\hspace{0.1cm}$|\hspace{0.9cm}...\hspace{0.3cm}|$\\
\hspace{-0.2cm}$\hspace{-0.1cm}\lambda\hspace{0.4cm}...$&$\hspace{-0.1cm}\lambda$&$\hspace{-0.2cm}\textcolor[rgb]{0.60,0.20,0.00}{a_1}$&$\hspace{-0.1cm}\lambda$&$\hspace{-0.1cm}\lambda\hspace{0.5cm}...$&$\hspace{-0.2cm}\lambda$&$\hspace{-0.1cm}\textcolor[rgb]{0.60,0.20,0.00}{a_2}$&$\hspace{-0.1cm}\lambda\hspace{0.1cm}...$&$\hspace{-0.2cm}\textcolor[rgb]{0.60,0.20,0.00}{a_{n_t-1}}$&$\hspace{-0.1cm}\lambda$&
$\lambda\hspace{1.4cm}\lambda$\\
\end{tabular}}

\begin{equation}
\hspace{0.8cm}\noindent \begin{tabular}{lllllllll}
$\hspace{-0.1cm}\textcolor[rgb]{0.60,0.20,0.90}{\large{[}^t_{j_{n_t}}}$&$\textcolor[rgb]{0.60,0.20,0.90}{\large{]}^t_{j_{n_t}}}$&$\hspace{0.2cm}\textcolor[rgb]{0.30,0.80,0.50}{]_{i_{k_{n_t}}}}\hspace{0.1cm}...\hspace{0.1cm}$&$\hspace{0.2cm}\textcolor[rgb]{0.30,0.80,0.50}{]_{i_{k_{n_t-1}}}}\hspace{0.1cm}...\hspace{0.1cm}$&$\hspace{-0.5cm}\textcolor[rgb]{0.20,0.10,0.90}{]_{i_{k_2}}\hspace{0.2cm}...}$&$\hspace{0.6cm}\textcolor[rgb]{1.00,0.50,0.00}{]_{i_{k_1}}\hspace{0.2cm}...}$&$\hspace{0.1cm}\textcolor[rgb]{1.00,0.50,0.00}{]_{i_1}}$\\
$|$&$|$&$\hspace{0.2cm}|\hspace{0.7cm}...\hspace{1cm}$&$\hspace{0.2cm}|\hspace{1.0cm}...\hspace{0.7cm}$&$\hspace{-0.5cm}|\hspace{0.6cm}...$&$\hspace{0.7cm}|\hspace{0.7cm}...$&$\hspace{0.1cm}|$\\
$\hspace{-0.2cm}\textcolor[rgb]{0.60,0.20,0.90}{a_{n_t}}$&$\hspace{-0.2cm}\textcolor[rgb]{0.60,0.20,0.90}{a_{n_t+1}}$&$\textcolor[rgb]{0.30,0.80,0.50}{a_{n_t+2}}\hspace{0.2cm}...$&$\hspace{-1.2cm}\textcolor[rgb]{0.30,0.80,0.50}{a_{n-k_{n_t-1}...-k_1+1}}\hspace{0.1cm}...\hspace{1cm}$&$\hspace{-1.4cm}\textcolor[rgb]{0.20,0.10,0.90}{a_{n-k_2-k_1+1}\hspace{0.1cm}}...\hspace{0.1cm}$&$\hspace{-0.1cm}\textcolor[rgb]{1.00,0.50,0.00}{a_{n-k_1+1}}\hspace{0.1cm}...$&$\textcolor[rgb]{1.00,0.50,0.00}{a_n}$\\
\end{tabular}
\label{p21}
\end{equation}
\vspace{0.1cm}
Next our aim is to find a connection between Theorem 2.9 and the Chomsky-Sch\"utzenberger theorem. More precisely we want 
to compute, from the structure of trace-words, the regular and the Dyck languages yielded by the Chomsky-Sch\"utzenberger theorem. 
Therefore, we build a transition-like diagram for context-free  grammars in Dyck normal form. First we build some directed graphs as follows.  
%\vspace{-0.1cm}
 
\begin{construction}\rm  
Let $G_k=(N_k, T, P_k, S)$ be an arbitrary context-free  grammar in Dyck normal form. A {\it dependency graph} of $G_k$ is a 
directed graph ${\cal G}^X=(V_X, E_X)$, $X \in \{]_j| [_j, ]_j\in N^{(3)}\}\cup \{S\}$, in which vertices 
are labeled with variables in $\tilde N_k\cup \{X\}$,  
$\tilde N_k=\{[_i| [_i, ]_i \in N^{(1)} \cup N^{(2)}_r\cup N^{(3)}\}\cup \{]_j| [_j, ]_j \in N^{(2)}_l\}$ and 
the set of edges is built as follows. For each rule $X\rightarrow [_i\hspace{0.1cm}]_i\in P_k$, $[_i, ]_i \in N^{(2)}_l$, 
${\cal G}^X$ contains a directed edge from $X$ to $]_i$, for each rule $X\rightarrow [_i\hspace{0.1cm}]_i\in P_k$, 
$[_i, ]_i \in N^{(1)}\cup N^{(2)}_r\cup N^{(3)}$, ${\cal G}^X$ contains a directed edge from $X$ to $[_i$. There exists an edge 
in ${\cal G}^X$ from a vertex labeled by $[_i$,  $[_i, ]_i \in N^{(2)}_r\cup N^{(3)}$, to a vertex labeled by $]_j$/$[_k$, 
$[_j, ]_j\hspace{-0.1cm}\in N^{(2)}_l$, $[_k, ]_k \in N^{(1)}\cup N^{(2)}_r\cup N^{(3)}$, if there exists a rule in 
$P_k$ of the form $[_i\rightarrow [_j\hspace{0.1cm}]_j$/$[_i\rightarrow [_k\hspace{0.1cm}]_k$. There exists an edge in ${\cal G}^X$ from 
a vertex labeled by $]_i$,  $[_i, ]_i \in N^{(2)}_l$, to a vertex labeled by $]_j$/$[_k$, $[_j, ]_j \in N^{(2)}_l$, 
$[_k, ]_k \in N^{(1)}\cup N^{(2)}_r\cup N^{(3)}$, if there exists a rule in $P_k$ of the form 
$]_i\rightarrow [_j\hspace{0.1cm}]_j$/$]_i\rightarrow [_k\hspace{0.1cm}]_k$. The vertex labeled by $X$ is called 
the {\it initial vertex} of ${\cal G}^X$. Any vertex labeled by a left bracket in $N^{(1)}$ is a {\it final vertex}.
\end{construction} 

Let ${\cal G}^X$ be a dependency graph of $G_k$. Consider the set of all possible paths in ${\cal G}^X$ starting 
from the initial vertex to a final vertex. Such a path is called {\it terminal path}. A {\it loop} or {\it cycle} in 
a graph is a path from $v$ to $v$ composed of distinct vertices. If from $v$ to $v$ there is no other vertex, then 
the loop is a {\it self-loop}. The {\it cycle rank} of a graph is a measure of the loop complexity formally defined\footnote{In brief, the rank of a cycle ${\cal C}$ is $1$ if there exists $v\in {\cal C}$ such that ${\cal C}-v$ 
is not a cycle. Recursively, the rank of a cycle ${\cal C}$ is $k$ if there exists $v\in {\cal C}$ such that 
${\cal C}-v$ contain a cycle of rank $k-1$ and all the other cycles in ${\cal C}-v$ have the rank at most $k-1$.} 
and studied in \cite{Coh} and \cite{E}. In \cite{E} it is proved that from each two vertices $u$ and $v$ belonging 
to a digraph of cycle rank $k$, there exists a regular expression of star-height\footnote{Informally, this is the (maximal) 
power of a nested $*$-loop occurring in the description of a regular expression. For the formal definition 
the reader is referred to \cite{E} and \cite{H1} (see also Definition 4.1, Section 4).} at most $k$ that describes the 
set of paths from $u$ to $v$. On the other hand, the cycle rank of a digraph with $n$ vertices is upper bounded 
by $n\log n$ \cite{GH}. Hence any regular expression obtained from a digraph with $n$ vertices has the star-height 
at most $n\log n$. Consequently, the (infinite) set of paths from an initial vertex to a final vertex in 
${\cal G}^X$, can be divided into a finite number of classes of terminal paths. Paths belonging to the same class 
are characterized by the same regular expression, in terms of $*$ and $+$ Kleene operations, of star-height at most $|V_X|\log |V_X|$ (which is finite related to the lengths of strings in $L(G_k)$). 
\vspace{0.1cm}

Denote by ${\cal R}^X_{[^t_i}$ the set of all regular expressions over $\tilde N_k\cup \{X\}$ that can be read in  
${\cal G}^X$, starting from the initial vertex $X$ and ending in the final vertex $[^t_i$. The cardinality of 
${\cal R}^X_{[^t_i}$ is finite. Define the homomorphism 
$h_{{\cal G}}\hspace{-0.1cm}: \tilde N_k\cup \{X\} \rightarrow \{]_i| [_i, ]_i \in N^{(2)}_r\cup N^{(3)}\}\cup \{\lambda\}$ 
such that $h_{{\cal G}}([_i)= ]_i$ for any $[_i, ]_i \in N^{(2)}_r\cup N^{(3)}$, $h_{{\cal G}}(X)=h_{{\cal G}}([^t_i)=h_{{\cal G}}(]_i)=\lambda$, for any  $[^t_i, ]^t_i \in N^{(1)}$ and $[_i, ]_i \in N^{(2)}_l$. For any element  $r.e^{(l,X)}_{[^t_i}\in {\cal R}^X_{[^t_i}$ we build a new 
regular expression\footnote{Since regular languages are closed under homomorphism and reverse operation,  
$r.e^{(r,X)}_{[^t_i}$ is a regular expression.} $r.e^{(r,X)}_{[^t_i}=h^r_{{\cal G}}(r.e^{(l,X)}_{[^t_i})$, 
where $h^r_{{\cal G}}$ is the mirror image of $h_{{\cal G}}$. Consider $r.e^X_{[^t_i}=r.e^{(l,X)}_{[^t_i}r.e^{(r,X)}_{[^t_i}$. For a certain $X$ and $[^t_i$, denote by 
${\cal R}.e^X_{[^t_i}$ the set of all regular expressions $r.e^X_{[^t_i}$ obtained as above. 
Furthermore,  ${\cal R}.e^X=\bigcup_{[^t_i, ]^t_i \in N^{(1)}}{\cal R}.e^X_{[^t_i}$ and 
${\cal R}.e={\cal R}.e^S\cup (\bigcup_{[_i, ]_i\in N^{(3)}}{\cal R}.e^{]_i}$). 
%\vspace*{-0.2cm}		
		
\begin{construction}\rm  
Let $G_k=(N_k, T, P_k, S)$ be a context-free  grammar in Dyck normal form and 
$\{{\cal G}^X| X \in \{]_j| [_j, ]_j \in N^{(3)}\cup \{S\}\}\}$ the set of dependency graphs of $G_k$. The 
{\it extended dependency graph} of $G_k$, denoted by ${\cal G}_e=({\cal V}_e, {\cal E}_e)$, is a directed graph 
for which ${\cal V}_e=\tilde N_k \cup \{S\}\cup\{]_i| [_i, ]_i \in N^{(2)}_r\cup N^{(3)}\}$, $S$ is the initial 
vertex of ${\cal G}_e$ and ${\cal E}_e$ is built as follows:\\
\vspace*{-0.4cm}

$1.$ - $S[_i$ ($S]_j$) - there exists an edge in ${\cal G}_e$ 
from the vertex labeled by $S$ to a vertex labeled by $[_i$ (from $S$ to $]_j$), $[_i, ]_i \in N^{(1)}\cup N^{(2)}_r\cup N^{(3)}$ 
($[_j, ]_j \in N^{(2)}_l$), if there exists a regular expression in ${\cal R}.e^S$ with a prefix of the form $S[_i$ ($S]_j$, respectively).      
\vspace*{0.1cm}

$2.$ - $]_i]_j$ - there exists an edge in ${\cal G}_e$ from a vertex labeled by 
$]_i$ to a vertex labeled by $]_j$ $[_i, ]_i, [_j, ]_j \in N^{(2)}_l$, if there exists a regular expression in ${\cal R}.e$ 
having a substring of the form $]_i]_j$ (if $i=j$ then $]_i]_i$ forms a self-loop in  ${\cal G}_e$).
\vspace*{0.1cm}

$3.$ - $]_i[_j$ (or $[_j]_i$) - there exists an edge in ${\cal G}_e$ 
from a vertex labeled by $]_i$ to a vertex labeled by $[_j$ (or vice versa from $[_j$ to $]_i$) such that $[_i, ]_i\in N^{(2)}_l$ 
and $[_j, ]_j \in N^{(2)}_r \cup N^{(3)}$, if there exists a regular expression in ${\cal R}.e$ having a substring of the form 
$]_i[_j$ ($[_j]_i$, respectively).
\vspace*{0.1cm}

$4.$ - $[_i[_j$ - there exists an edge in ${\cal G}_e$ from a vertex labeled by $[_i$ to a vertex 
labeled by $[_j$, $[_i, ]_i, [_j, ]_j\in N^{(2)}_r\cup N^{(3)}$, if there exists a regular expression in ${\cal R}.e$ having a 
substring of the form $[_i[_j$ (if $i=j$ then $[_i[_i$ forms a self-loop in  ${\cal G}_e$).
\vspace*{0.1cm}

$5.$ - $]_i[^t_j$ (or $[_i[^t_j$) - there exists an edge in ${\cal G}_e$ 
from a vertex labeled by $]_i$ (or by $[_i$) to a vertex labeled by $[^t_j$, $[_i, ]_i \in N^{(2)}_l$ (or $[_i, ]_i \in N^{(2)}_r$, respectively),  
$[^t_j, ]^t_j \in N^{(1)}$, if there exists a regular expression in ${\cal R}.e$ with a substring of the form $]_i[^t_j$ ($[_i[^t_j$, respectively).   
\vspace*{0.1cm}

$6.$ - $]_j[^t_i$ - there exists an edge in ${\cal G}_e$ from a vertex labeled by $]_j$ to a vertex labeled 
by $[^t_i$,  $[_j, ]_j \in N^{(3)}$, $[^t_i, ]^t_i \in N^{(1)}$, if there exists a regular expression in ${\cal R}.e^{]_j}_{[^t_i}$ of the form $]_j[^t_i$. 
\vspace*{0.1cm}

$7.$ - $]_j[_i$ (or $]_j]_i$) - there exists an edge in ${\cal G}_e$ from a 
vertex labeled by $]_j$ to a vertex labeled by $[_i$, $[_j, ]_j \in N^{(3)}$, $[_i, ]_i \in N^{(2)}_r$ ($[_i, ]_i \in N^{(2)}_l$, respectively), 
if there exists a regular expression in ${\cal R}.e^{]_j}$ with a prefix of the form $]_j[_i$ ($]_j]_i$, respectively).  
\vspace*{0.1cm}

$8.$ - $]_i]_j$ - there exists an edge in ${\cal G}_e$ from a vertex  $]_i$ to a vertex labeled by 
$]_j$, $i$ and $j$ not necessarily distinct, such that $[_i, ]_i\in N^{(2)}_r$, $[_j, ]_j\in N^{(2)}_r\cup N^{(3)}$, if either 
$i.$, $ii.$, or $iii.$ holds:\\
\hspace*{0.3cm} $i.$ there exists at least one regular expression in ${\cal R}.e$ having a substring of the form $]_i]_j$ (if $i=j$ then 
$]_i]_i$ forms a self-loop in  ${\cal G}_e$),\\ 
\hspace*{0.3cm} $ii.$ there exists $[_k, ]_k \in N^{(3)}$ such that there exist a regular expression in ${\cal R}.e$ with a substring of 
the form $]_k]_j$, and a regular expression in ${\cal R}.e^{]_k}$ that ends in $]_i$ (if $i=j$ then $]_i]_i$ is a self-loop). \\
\hspace*{0.3cm} $iii.$ there exist $[_k, ]_k,[_{k_1}, ]_{k_1}, ...,[_{k_m}, ]_{k_m} \in N^{(3)}$ such that there exist   
a regular expression in ${\cal R}.e$ with a substring of the form $]_k]_j$, a regular expression in ${\cal R}.e^{]_k}$ that ends in $]_{k_1}$, a 
regular expression in ${\cal R}.e^{]_{k_1}}$ that ends in $]_{k_2}$, and so on, until a regular expression in ${\cal R}.e^{]_{k_{m-1}}}$ 
ending in $]_{k_m}$ and a regular expression in ${\cal R}.e^{]_{k_m}}$ ending in $]_i$ are reached.
\vspace*{0.1cm}

$9.$ - $[^t_i]_j$ - there exists an edge in ${\cal G}_e$ from a vertex labeled by $[^t_i$,  
$[^t_i, ]^t_i \in N^{(1)}$, to a vertex labeled by $]_j$, $[_j, ]_j\in N^{(2)}_r\cup N^{(3)}$ if either $i.$, $ii.$, or $iii.$ holds\\
\hspace*{0.3cm} $i.$ there exists a regular expression in ${\cal R}.e$ having a substring of the form $[^t_i]_j$,\\ 
\hspace*{0.3cm} $ii.$ there exists $[_k, ]_k \in N^{(3)}$ such that there exist a regular expression in ${\cal R}.e$ having a 
substring of the form $]_k]_j$, and a regular expression in ${\cal R}.e^{]_k}_{[^t_i}$ that ends in $[^t_i$. \\
\hspace*{0.3cm} $iii.$ there exist $[_k, ]_k,[_{k_1}, ]_{k_1}, ...,[_{k_m}, ]_{k_m} \in N^{(3)}$  such that there exist 
a regular expression in ${\cal R}.e$ with a substring of the form $]_k]_j$, a regular expression in ${\cal R}.e^{]_k}$ that ends  
in $]_{k_1}$, a regular expression in ${\cal R}.e^{]_{k_1}}$ that ends in $]_{k_2}$, and so on, until a regular expression 
in ${\cal R}.e^{]_{k_{m-1}}}$ ending in $]_{k_m}$ and a regular expression in ${\cal R}.e^{]_{k_m}}_{[^t_i}$ ending in $[^t_i$ are reached.
\vspace*{0.1cm}

\hspace*{-0.2cm}$10.$ -  A vertex labeled by $[^t_i$, $[^t_i, ]^t_i \in N^{(1)}$, is a final vertex in 
${\cal G}_e$ if either $i.$, $ii.$, or $iii.$ holds: \\
\hspace*{0.3cm} $i.$ there exists a regular expression in ${\cal R}.e^S$ that ends in $[^t_i$, \\
\hspace*{0.3cm} $ii.$ there exists $[_k, ]_k \in N^{(3)}$, such that there exist a regular 
expression in ${\cal R}.e^S$ that ends in $]_k$, and a regular expression in ${\cal R}.e^{]_k}_{[^t_i}$ that ends in $[^t_i$.\\ 
\hspace*{0.3cm} $iii.$ there exists $[_k, ]_k \in N^{(3)}$ such that there exist a regular expression in ${\cal R}.e^S$ that ends  
in $]_k$, and $[_{k_1}, ]_{k_1}, ...,[_{k_m}, ]_{k_m} \hspace*{-0.1cm}\in N^{(3)}$  such that there is a regular expression in 
${\cal R}.e^{]_k}$ that ends in $]_{k_1}$, a  regular expression in ${\cal R}.e^{]_{k_1}}$ that ends in $]_{k_2}$, and so on, until 
a  regular expression in ${\cal R}.e^{]_{k_{m-1}}}$ ending in $]_{k_m}$ and a  regular expression in ${\cal R}.e^{]_{k_m}}$ ending in $[^t_i$ are reached.
\vspace*{0.1cm}

\hspace*{-0.2cm}$11.$ -  A vertex labeled by $]_i$, $[_i, ]_i \in N^{(2)}_r$,  is a final vertex in ${\cal G}_e$ if either 
$i.$, $ii.$, or $iii.$ holds:\\
\hspace*{0.3cm} $i.$ there exists a regular expression in ${\cal R}.e^S$ that ends in $]_i$, \\
\hspace*{0.3cm} $ii.$ there exists $[_k, ]_k \in N^{(3)}$, such that there exist a regular expression in 
${\cal R}.e^S$ that ends  in $]_k$, and a regular expression in ${\cal R}.e^{]_k}$ that ends in $]_i$.\\
\hspace*{0.3cm} $iii.$ there exists $[_k, ]_k \in N^{(3)}$ such that there exist a regular expression in 
${\cal R}.e^S$ that ends  in $]_k$, and $[_{k_1}, ]_{k_1}, ...,[_{k_m}, ]_{k_m} \in N^{(3)}$  such that 
there exist a regular expression in ${\cal R}.e^{]_k}$ ending in $]_{k_1}$, a regular expression in 
${\cal R}.e^{]_{k_1}}$ ending in $]_{k_2}$, and so on, until a regular expression in ${\cal R}.e^{]_{k_{m-1}}}$ 
ending in $]_{k_m}$ and a regular expression in ${\cal R}.e^{]_{k_m}}$ ending in $]_i$ are reached. 
\end{construction}

Denote by ${\cal R}_e$ the set of all regular expressions obtained by reading all paths in ${\cal G}_e$ from 
the initial vertex $S$ to all final vertices (i.e., all terminal paths). We have  

\begin{theorem}(Chomsky-Sch\"utzenberger theorem)\rm \hspace{0.1cm} For each context-free  language $L$ there 
exist an integer $K$, a regular set $R$, and a homomorphism $h$, such that $L=h(D_K\cap R)$. Furthermore, if $G$ 
is the context-free  grammar that generates $L$, $G_k$ the Dyck normal form of $G$, and $G_k$ has no extended grammar, then $K=k$ and 
$D_K\cap R=\L(G_k)$. Otherwise, there exists $p>0$ such that $K=k+p$, and $D_K\cap R=D'_K$, where $D'_K$ is the 
subset of $D_K$ computed as in Theorem 2.9. 
\end{theorem}

\begin{proof}
Let $G_k=(N_k, T, P_k, S)$ be the Dyck normal form of $G$ such that $L=L(G)$. Suppose that $G_k$ does not have an extended 
grammar. Let $h_k\hspace{-0.1cm}: \tilde N_k\cup \{]_i| [_i,]_i \in N^{(2)}_r\cup N^{(3)}\}\cup \{S\}\rightarrow \{[_i,]_i| [_i,]_i \in N^{(2)}_r\cup N^{(3)}\}\cup\{[_i]_i| [_i,]_i \in N^{(2)}_l\cup N^{(1)}\}\cup \{\lambda\}$ be the homomorphism defined by $h_k(S)=\lambda$, $h_k([_i)= [_i$, $h_k(]_i)=\hspace{0.1cm}]_i$ 
for $[_i, ]_i \in N^{(2)}_r\cup N^{(3)}$, $h_k(]_i)=[_i\hspace{0.1cm}]_i$ for $[_i, ]_i \in N^{(2)}_l$, 
and $h_k([^t_i)=[^t_i\hspace{0.1cm}]^t_i$ for $[^t_i, ]^t_i \in N^{(1)}$. Then $R=h_k({\cal R}_e)$ is a 
regular language such that $D_k\cap R=\L(G_k)$. 

To prove the last equality, notice that each terminal path in a dependency graph ${\cal G}^X$ (Construction 3.3) provides a string equal to a substring (or a prefix if $X=S$) of a trace-word in $\L(G_k)$ (in which left brackets in $N^{(2)}_l$ are omitted) generated (in the leftmost derivation order) from the derivation time when $X$ is rewritten, up to the moment when the very first left bracket of a pair in $N^{(1)}$ is rewritten. This string corresponds to a regular expression $r.e^{(l,X)}_{[^t_i}\in {\cal R}^{X}_{[^t_i}$, which is extended with another regular 
expression $r.e^{(r,X)}_{[^t_i}$ that is the ``mirror image'' of left brackets in $N^{(2)}_r\cup N^{(3)}$ occurring 
in $r.e^{(l,X)}_{[^t_i}$. If left brackets in $N^{(2)}_r\cup N^{(3)}$ are enrolled in a star-height, then their homomorphic image (through $h^r_{\cal G}$) in $r.e^{(r,X)}_{[^t_i}$ is another star-height. The ``mirror image'' of consecutive left brackets in $N^{(2)}_r$ (with respect to their relative core) is a segment composed of consecutive right brackets in $N^{(2)}_r$. The ``mirror image'' of consecutive left brackets in $N^{(3)}$ is ``broken'' by the 
interpolation of a regular expression $r.e^{]_j}_{[^t_i}$ in ${\cal R}.e^{]_j}$, $[_j, ]_j \in N^{(3)}$. The 
number of $r.e^{]_j}_{[^t_i}$ insertions matches the number of left brackets $[_j$ placed at the left side of the relative core (this is assured by the intersection with $D_k$). In fact, the  extended dependency graph of $G_k$ 
has been conceived such that it reproduces, on regular expressions in ${\cal R}_e$, the structure of trace-words in $\L(G_k)$. The main problem is the ``star-height synchronizations'' for brackets in $N^{(2)}_r\cup N^{(3)}$, i.e., 
the number of left-brackets occurring in a loop placed at the left-side of a core segment $[^t_i \hspace{0.1cm}]^t_i$, to be equal to 
the number of their pairwise right-brackets occurring in the corresponding ``mirror'' loop placed at the right-side of its relative core, 
$[^t_i \hspace{0.1cm}]^t_i$, $[^t_i\hspace{0.1cm}\in N^{(1)}$. 
This is controlled by the intersection of $h_k({\cal R}_e)$ with $D_k$, leading to $\L(G_k)$. In 
few words, the proof is by the construction described in Construction 3.4.  Another problem that occurs is that 
the construction of ${\cal G}_e$ allows to concatenate $r.e^{(l,X)}_{[^t_i}\in {\cal R}^{X}_{[^t_i}$ to its right pairwise $r.e^{(r,X)}_{[^t_i}$ 
as well as to another regular expression $\overline{r.e}^{(r,X')}_{[^t_i}$ (which by construction it is also concatenated to its left pairwise $\overline{r.e}^{(l,X')}_{[^t_i}$) where $X$ and $X'$ 
are not necessarily distinct. This does not change the intersection with the Dyck language, but enlarges the regular language $R=h_k({\cal R}_e)$ with useless\footnote{In Section 4 we show how these unnecessary  concatenations can be avoided, through 
a refinement procedure of the regular language in the Chomsky-Sch\"utzenberger theorem.} words.  
%The intersection of $D_k$ with a regular expression. in ${\cal R}_e$ makes the synchronization of star-heights, composed of left brackets in 
%$N^{(2)}_r\cup N^{(3)}$, occurring at the left side of a core segment $[^t_i]^t_i$) with their star-heights counterparts, composed of right 
%brackets in $N^{(2)}_r \cup N^{(3)}$, placed at the right side of $[^t_i]^t_i$). 
  
%Thus, keeping $r.e^{(l,X)}_{[^t_i}$ as correct\footnote{It is indeed correct, because $h_k(r.e^{(l,X)}_{[^t_i})$ is a correct substring of a 
%trace-word, generated from the derivation time when $X$ is rewritten, up to the moment when $[^t_i$ is rewritten.} the right ``extension'' 
%from $r.e^{(r)}_{[^t_i}$, that must be added to $r.e^{(l)}_{[^t_i}$, to form another correct substring of a trace-word, is chosen by intersecting %$h_K(r.e^{(r)}_{[^t_i}h^r_{{\cal G}}(r.e^{(l)}_{[^t_i}))$ with $D_k$. If we consider all expressions in ${\cal R}_e$, we obtain 
%$\L_k(G_k)= D_k \cap h_k({\cal R}_e)$. 

If $G_k$ has an extended grammar $G_{k+p}\hspace{-0.1cm}=(N_{k+p}, T, P_{k+p}, S)$, built as in the proof of 
Theorem 2.9, then ${\cal R}_e$ is augmented with ${\cal r}_e\hspace{-0.1cm}=\hspace{-0.1cm}\{S[_{t_{k+1}}, 
..., S[_{t_{k+p}}\}$ and $h_k$ is extended to  $h_K\hspace{-0.1cm}: \tilde N_k \cup \{S\}\cup \{]_i| [_i,]_i 
\in N^{(2)}_r\cup N^{(3)}\}\cup\{[_{t_{k+1}}, ..., [_{t_{k+p}}\} \hspace{-0.1cm}\rightarrow \hspace{-0.1cm}\{[_i,]_i| [_i,]_i \in N^{(2)}_r\cup N^{(3)}\}\cup \{[_i]_i| [_i,]_i\hspace{-0.1cm}\in N^{(2)}_l\cup N^{(1)}\} \cup \{[_{t_{k+1}}]_{t_{k+1}},$ $ ..., [_{t_{k+p}}]_{t_{k+p}}\}\cup \{\lambda\}$, 
where $h_K(x)=h_k(x)$, $x\notin \{[_{t_{k+1}}, ..., [_{t_{k+p}}\}$, $h_K([_{t_{k+j}})=[_{t_{k+j}}]_{t_{k+j}}$, $1\leq j\leq p$, $K=k+p$. $\L(G_k)$ 
is augmented with $L_p=\{[_{t_{k+1}} ]_{t_{k+1}}, ..., [_{t_{k+p}}]_{t_{k+p}}\}$ and 
$D'_K=h_K({\cal R}_e\cup {\cal r}_e)\cap D_K =\L(G_k)\cup L_p$. 
\vspace{0.1cm}

The homomorphism $h$ is equal to $\varphi$ in 
Theorem 2.9, i.e., $\varphi\hspace{-0.1cm}: (N_{k+p}- \{S\})^*\rightarrow T^*$, $\varphi(N)=\lambda$, 
for each rule of the form $N\rightarrow XY$, $N, X, Y\in N_k$, and $\varphi(N)=t$, for each rule of the form 
$N\rightarrow t$, $N\in N_k-\{S\}$, $t\in T$, $\varphi([_{k+i})=t_{k+i}$, and $\varphi(]_{k+i})=\lambda$, for 
each $1\leq i \leq p$.
\end{proof}
\vspace{0.1cm}

Note that, for the case of linear languages there is only one dependency graph ${\cal G}^S$. The regular 
language in the Chomsky-Sch\"utzenberger theorem can be built without the use of the extended dependency graph. It suffices 
to consider only the regular expressions in ${\cal R}.e^S=\bigcup_{[^t_i, ]^t_i \in N^{(1)}}{\cal R}.e^S_{[^t_i}$. If $G_k$ has an extended grammar $G_K$, then $L(G_k)= \varphi(D_K\cap h_K({\cal R}.e^S\cup {\cal r}_e))$, where 
$K=k+p$, $G_K$, ${\cal r}_e$, and $\varphi$ are defined as in Theorems 2.9 and 3.5. If $G_k$ has no extended 
grammar then $L(G_k)= \varphi(D_k\cap h_k({\cal R}.e^S))$. However, a graphical representation may be considered 
an interesting common framework for both, linear and context-free  languages. Below we illustrate the manner in which the 
regular language in the Chomsky-Sch\"utzenberger theorem can be computed for linear (Examples 3.6) and context-free  (Example 3.7) languages. 
%%%%%%%
%The interpretation is simple and natural. The regular language in the Chomsky-Sch\"utzenberger 
%theorem intersected with a (certain) Dyck language lists all derivation trees (read in the depth-first search order)  
%associated with words in a context-free  grammar, in Dyck normal form or in Chomsky normal form (since these derivation trees are equal, up to an homomorphism). 
%The intersection forms (with very little exceptions) the trace-languages associated with the respective context-free  grammar. 
%%%%%%%
%The interpretation is simple and natural. The regular language, certified by the Chomsky-Sch\"utzenberger 
%theorem intersected with a (certain) Dyck language list all derivations obtained by reading, in a depth-first 
%search order, the derivation trees associated with words in a context-free G either in Dyck normal form or in Chomsky normal form (since the derivation 
%trees are equal, up to an homomorphism). The intersection forms (with very little exceptions) the trace-languages 
%associated with that context-free G. 
%Examples ... illustrate the manner in which the regular language and the homomorphism in the 
%Chomsky-Sch\"utzenberger theorem can be built for the case of LIN languages. 
%\vspace{-0.3cm}
\vspace*{-0.1cm}

\begin{example}\rm
Consider the linear context-free  grammar  $G=( \{S,[_1...,[_7, ]_1...,]_7\}, \{ a,b,c,d\},\\ S, P)$ in linear-Dyck normal form, with  
$P\hspace{-0.1cm}=\{ S\rightarrow [^t_1\hspace{0.1cm}]_1,  ]_1\hspace{-0.1cm}\rightarrow [_2\hspace{0.1cm}]^t_2,[_2\rightarrow [_3\hspace{0.1cm}]^t_3, 
[_3\rightarrow [_2\hspace{0.1cm}]^t_2/[^t_4\hspace{0.1cm}]_4,]_4\rightarrow [_5\hspace{0.1cm}]^t_5, 
[_5\rightarrow [^t_6\hspace{0.1cm}]_6, ]_6\rightarrow [^t_1\hspace{0.1cm}]_1/[^t_7\hspace{0.1cm}]_7^t, [^t_1
\rightarrow a, ]^t_2 \rightarrow b,]^t_3 \rightarrow c, [^t_4 \rightarrow b, ]^t_5 \rightarrow d,
[^t_6 \rightarrow b, [^t_7 \rightarrow a, ]^t_7 \rightarrow a\}$.
\vspace*{-0.3cm}\\

\hspace{-0.1cm}The dependency graph ${\cal G}^S$ and extended dependency graph ${\cal G}_e$ of $G$ are depicted in 
Figu-\\re 1.a and 1.b, respectively. There exists only one regular expression readable from ${\cal G}^S$, i.e.,  
$r.e^{(l,S)}_{[^t_7}\hspace{-0.1cm}=S(]_1([_2[_3)^+]_4[_5]_6)^+[^t_7$. Hence, $r.e^{S}_{[^t_7}=r.e^{(l,S)}_{[^t_7}r.e^{(r,S)}_{[^t_7}\hspace{-0.1cm}=S(]_1([_2[_3)^+]_4
[_5]_6)^+[^t_7(]_5(]_3]_2)^+)^+$.

The regular language provided by the Chomsky-Sch\"utzenberger theorem is 

\centerline{$R=([_1\hspace{0.1cm}]_1([_2\hspace{0.1cm}[_3)^+[_4\hspace{0.1cm}]_4\hspace{0.1cm}
[_5\hspace{0.1cm}[_6\hspace{0.1cm}]_6)^+[^t_7]^t_7(]_5\hspace{0.1cm}(]_3\hspace{0.1cm}]_2)^+)^+$.}

Therefore, $D'_7= D_7\cap R =\{([_1\hspace{0.1cm}]_1([_2\hspace{0.1cm}[_3)^n[_4\hspace{0.1cm}]_4
\hspace{0.1cm}[_5\hspace{0.1cm}[_6\hspace{0.1cm}]_6)^m[^t_7]^t_7(]_5\hspace{0.1cm}(]_3
\hspace{0.1cm}]_2)^n)^m|n,m\geq 1\}=\L(G_k)$, and $L(G)= \varphi(D'_7)=\hspace{-0.1cm}\{(abb)^maa(d(cb)^n)^m|n,m\geq 1\}$ 
($G$ contains no rule of the form $S\rightarrow t$, $t\in T$). 

\begin{figure}
	\centering
		\includegraphics[scale=0.56]{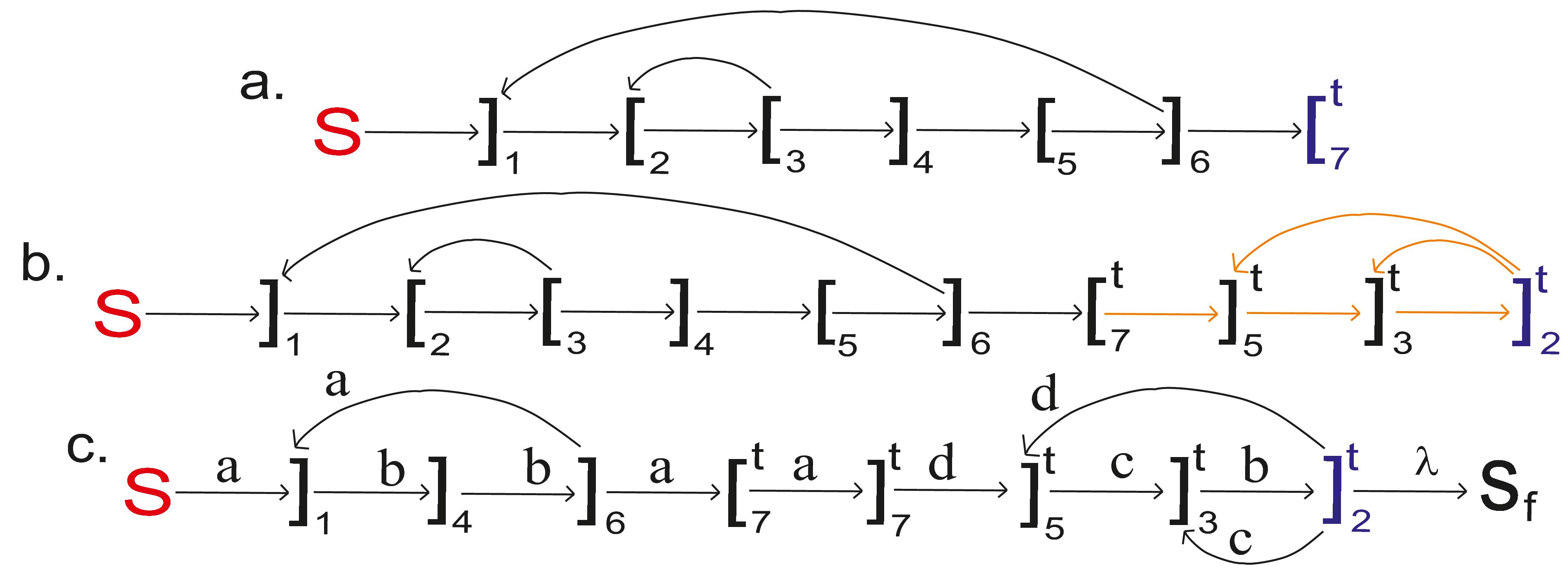}
		\vspace{-0.3cm}
		\caption{{\footnotesize a. The dependency graph ${\cal G}^S$ of grammar $G$ in Example 1. b. The 
        extended dependency graph of $G$. Edges colored in orange extend $\cal G$ to ${\cal G}_e$. c. The 
        transition diagram ${\cal A}_e$ (see Example 5.1 a.) built from ${\cal G}_e$. Each bracket $[_i$ 
        ($S$, $]_i$) in ${\cal A}_e$ corresponds to state $s_{[_i}$ ($s_S$, $s_{]_i}$). In all graphs $S$ 
        is the initial vertex. In a. - b. the vertex colored in blue is the final vertex.}}
	%\label{fig21}
\end{figure}
\end{example}
\vspace{-0.3cm}
\begin{example}\rm
Consider the context-free  grammar $G=( \{S,[_1...,[_7, ]_1...,]_7\}, \{ a,b,c\}, S, P)$ in Dyck normal form with  
$P=\{ S\rightarrow [_1]_1,[_1\rightarrow[_5]_5^t/[_1]_1, ]_1\rightarrow [_6]_6, [_2\rightarrow [_6]_6/[^t_7]_7, 
[_3\rightarrow [^t_7]_7, [_5\rightarrow [^t_4]^t_4, [_6 \rightarrow [_3]^t_3, ]_6 \rightarrow [_2]^t_2,
]_7\rightarrow [_3]^t_3/[^t_4]^t_4,]^t_2\rightarrow b, ]^t_3 \rightarrow a, [^t_4\rightarrow c, ]^t_4\rightarrow c, 
]^t_5\rightarrow b, [^t_7\rightarrow a\}$
\vspace{0.1cm}

The sets of regular expressions and extended regular expressions obtained by reading ${\cal G}^S$ (Figure 2.a) are  
${\cal R}^S_{[^t_4}=\{S[^+_1[_5[^t_4\}$ and ${\cal R}.e^S={\cal R}.e^S_{[^t_4}=\{S[^+_1[_5[^t_4]^t_5]^+_1\}$, respectively. 

The regular expressions and extended regular expressions readable from ${\cal G}^{]_1}$ (Figure 2.b) are 
${\cal R}^{]_1}_{[^t_4}=\{]_1[_6([_3]_7)^+[^t_4\}$ and ${\cal R}.e^{]_1}=\{]_1[_6([_3]_7)^+[^t_4(]^t_3)^+]_6\}$, respectively.
The regular expressions and extended regular expressions obtained by reading ${\cal G}^{]_6}$ (Figure 2.c) are
${\cal R}^{]_6}_{[^t_4}\hspace{-0.1cm}=\{]_6[_2[_6([_3]_7)^+[^t_4, ]_6[_2(]_7[_3)^*]_7[^t_4\}$ and  
${\cal R}.e^{]_6}\hspace{-0.1cm}={\cal R}.e^{]_6}_{[^t_4}\hspace{-0.1cm}=\{]_6[_2[_6([_3]_7)^+[^t_4(]^t_3)^+]_6]^t_2, 
]_6[_2(]_7[_3)^*]_7[^t_4(]^t_3)^*]^t_2\}$, respectively. 
%\vspace{-0.1cm}

The extended dependency graph of $G$ is sketched in Figure 2.d. Edges in black, are built from the regular 
expressions in ${\cal R}^X_{[^t_4}$, $X\in\{S, ]_1,]_6\}$.  Orange edges emphasize symmetrical structures, built 
with respect to the structure of trace-words in $\L(G)$. Some of them (e.g., $]^t_2]_1$ and $]^t_2]^t_2$) connect 
regular expressions in ${\cal R}_e$ between them with respect to the structure of trace-words in $\L(G)$ (see Construction 3.4, item 8). 
The edge $]^t_2]_1$ is added because there exists at least one regular expression in ${\cal R}_e$ that contains $]_1]_1$ (e.g. $S[^+_1[_5[^t_4]^t_5]^+_1$), 
a regular expression in ${\cal R}.e^{]_1}_{[^t_4}$ that ends in $]_6$ (e.g. $]_1[_6([_3]_7)^+[^t_4(]^t_3)^+]_6$) and a regular expression in 
${\cal R}.e^{]_6}_{[^t_4}$ that ends in $]^t_2$ (see Construction 3.4, item $8.iii.$).  The + self-loop 
$]^t_2]^t_2$ is due to the existence of a regular expression that contains $]_6]^t_2$ (e.g. $]_6[_2[_6([_3]_7)^+[^t_4(]^t_3)^+]_6]^t_2$) and a 
regular expression in ${\cal R}.e^{]_6}_{[^t_4}$ that ends in $]^t_2$ (e.g. $]_6[_2[_6([_3]_7)^+[^t_4(]^t_3)^+]_6]^t_2$ or $]_6[_2(]_7[_3)^*]_7[^t_4(]^t_3)^*]^t_2$). 
%\vspace*{0.1cm}

The regular language provided by the Chomsky-Sch\"utzenberger theorem is the homomorphic image, through $h_k$ (defined in Theorem 3.5), 
of all regular expressions associated with all paths in the extended dependency graph in Figure 2.d, reachable 
from the initial vertex $S$ to the final vertex labeled by $]^t_2$, i.e., terminal paths.  
\vspace*{-0.1cm}

\begin{figure}
	\centering
		\includegraphics[scale=0.71]{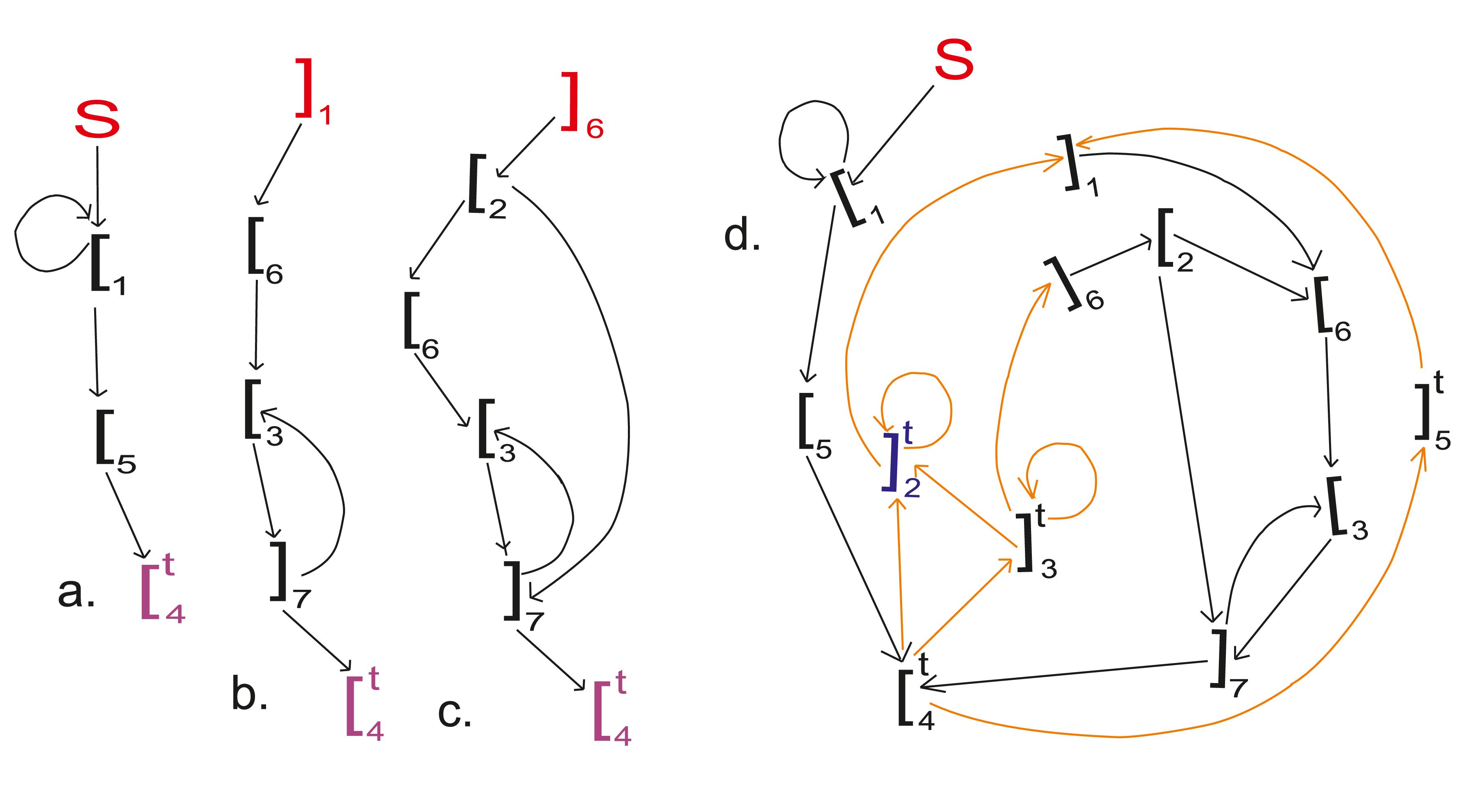}
		\vspace{-0.3cm}
	\caption{{\footnotesize a. - d. The dependency graphs of the context-free  grammar $G$ in Example 3.7. e. The extended 
      dependency graph of $G$. In all graphs,  vertices colored in red are initial vertices, while vertices colored 
      in blue are final vertices. Edges colored in orange, in $d.$ emphasize symmetrical structures obtained by 
      linking the dependency graphs between them.}}
\end{figure}
\end{example}
%\vspace*{-0.1cm}

The interpretation that emerges from the graphical method described in this paper is that {\it the regular language in 
the Chomsky-Sch\"utzenberger theorem intersected with a (certain) Dyck language lists all derivation trees (read in 
the depth-first search order) associated with words in a context-free  grammar, in Dyck normal form or in Chomsky normal 
form (since these derivation trees are equal, up to an homomorphism). The intersection forms (with very little exceptions) 
the trace-language associated with the respective context-free  grammar.} 

%Each method \cite{B}, \cite{G}, \cite{H} used to prove the Chomsky-Sch\"utzenberger theorem  has its own strategy to build a regular and a Dyck 
%languages (such that the intersection through a certain homomorphism reveals a context-free  language). Consequently, for a certain 
%context-free  language there may exist infinitely many regular languages provided by this representation theorem. The questions are: 
%{\it 1. For a certain context-free  language can we build a minimal regular language that satisfy the Chomsky-Sch\"utzenberger theorem 2. Is this useful 
%in further studies of context-free  languages?} In the next section we conjecture that the answer to the first question is negative. 
%\vspace*{0.1cm}

In the next section we refine the extended dependency graph ${\cal G}_e$ to provide a thiner regular language in the Chomsky-Sch\"utzenberger 
theorem with respect to the structure of the context-free grammar in Dyck normal form obtained through the algorithm described in the proof of 
Theorem 1.2. Based on this readjustment in Section 5 we sketch a {\it transition diagram} for a {\it finite automaton} and a {\it regular grammar} 
that generates a \textit{regular superset approximation} for the initial context-free language. 

\section{Further Refinements of the Regular Language in the\\Chomsky-Sch\"utzenberger Theorem}

One of the main disadvantage of considering $*$-height regular expressions in building the extended dependency 
graph associated with a context-free grammar in Dyck normal form is that some $*$-loops composed of right brackets 
in $N^{(2)}_r\cup N^{(3)}$ may not be symmetrically arranged according to their corresponding left brackets in 
$N^{(2)}_r\cup N^{(3)}$, if we consider their corresponding core segment as a symmetrical center. This is due to  
the possibility of having ``$\lambda$-loops''. This deficiency does not affect the intersection with a Dyck 
language, but it has the disadvantage of enlarging considerable the regular language  in the Chomsky-Sch\"utzenberger 
theorem. This can be avoided by considering only loops described in terms of $+$ Kleene closure. 
%when concatenating the regular expression associated with a path in a dependency graph that 
%correspond to the first half of a trace-word with its mirror image through $h_k$ is that due to $*$-loops of right brackets in 
%$N^{(2)}_r\cup N^{(3)}}$ may not be symmetrically arranged with respect to their corresponding core segment.   
%To avoid unsymmetrical regular expressions that can be obtained by using $*$-loops    
%\vspace{-0.1cm}

Another disfunction of the extended dependency graph built through  Construction 3.4 is the concatenation of a regular 
expression $r.e^{(l,X)}_{[^t_i}$ with another regular expression $\overline{r.e}^{(r,X')}_{[^t_i}$, 
$\overline{r.e}^{(r,X')}_{[^t_i}\neq r.e^{(r,X)}_{[^t_i}$ (due to the common 
tie $[^t_i$ that marks a core segment). This can be avoided by a renaming procedure of the regular expressions 
we want to concatenate. All these additional modifications in building an extended dependency graph are useful 
only if we want to refine the regular language that satisfies the Chomsky-Sch\"utzenberger theorem (with regards 
to the grammar in Dyck normal form). This will be further handled (in Section 5) to build a tighter approximation 
for the context-free language it characterizes.      
\vspace{0.1cm}

Each regular expression of a finite star-height can be described as a finite union of regular expressions in 
terms of $+$ Kleene closure (shortly plus-height). For instance the $*$-height regular expression 
$]_6[_2(]_7[_3)^*]_7[^t_4$  in ${\cal R}^{]_6}_{[^t_4}$ can be forked into $]_6[_2(]_7[_3)^+]_7[^t_4$ and $]_6[_2]_7[^t_4$. 
The plus-height of a regular expression, can be defined analogous to the star-height of a regular expression in \cite{H1}, as follows.     
\vspace{-0.1cm}

\begin{definition}\rm
Let $\Sigma$ be a finite alphabet. The plus-height $h(r)$ of a regular expression $r$ is defined 
recursively as follows: $i.$ $h(\lambda)=h(\emptyset)=h(a)=0$ for $a\in \Sigma$, 
$ii.$ $h(r_1\cup r_2)=h(r_1r_2)=$ max$\{h(r_1), h(r_2)\}$, and $h(r^+)=h(r)+1$. 
\end{definition}
%\vspace{-0.1cm}
   
Note that for any star-height regular expression it is possible to build a digraph, with an initial vertex $v_i$ 
and a final vertex $v_f$, such that all paths in this digraph, from $v_i$ to $v_f$, to provide the respective 
regular expression (which can be done in a similar manner as in Construction 3.4). However, if the regular 
expression is described in terms of plus-height then this statement may not be true (due to the repetition of 
some symbols). To force this statement be true, also for plus-height regular expressions, each repetition of a 
bracket is marked by a distinct symbol (e.g., $]_6[_2(]_7[_3)^+]_7[^t_4$ becomes $]_6[_2(]_7[_3)^+\bar ]_7[^t_4$), 
and then, for the new plus-height regular expression obtained in this way,  we build a digraph with the above 
property. In order to recover the initial plus-height regular expression from the associated digraph, a homomorphism 
that maps all the marked brackets (by distinct symbols) into the initial one must be applied. Each time it is 
required, we refer to such a vertex as a {\it $\hbar$-marked} vertex. Therefore, due to the technical transformations 
described above and the symmetrical considerations used in the construction of a trace language, we may assume to work 
only with plus-height regular expressions. 
\vspace{0.1cm}

Let $G_k=(N_k, T, P_k, S)$ be an arbitrary context-free  grammar in Dyck normal form, and ${\cal G}^X$ the dependency graph of $G_k$ 
(see Construction 3.3). Denote by ${\cal P}^X_{[^t_i}$ the set of all plus-height regular expressions over 
$\tilde N_k\cup \{X\}$ that can be read in ${\cal G}^X$, starting from the initial vertex $X$ and ending in the 
final vertex $[^t_i$. The cardinality of ${\cal P}^X_{[^t_i}$ is finite. Now, we consider the same homomorphism, 
as defined for the case of the set ${\cal R}^X_{[^t_i}$, i.e., $h_{{\cal G}}\hspace{-0.1cm}: \tilde N_k\cup \{X\} \rightarrow 
\{]_i| [_i, ]_i \in N^{(2)}_r\cup N^{(3)}\}\cup \{\lambda\}$ such that $h_{{\cal G}}([_i)= ]_i$ for any pair 
$[_i, ]_i \in N^{(2)}_r\cup N^{(3)}$, $h_{{\cal G}}(X)=h_{{\cal G}}([^t_i)=h_{{\cal G}}(]_i)=\lambda$, 
for any  $[^t_i, ]^t_i \in N^{(1)}$ and $[_i, ]_i \in N^{(2)}_l$. For any element  $r.e^{(l,X)}_{[^t_i}\in {\cal P}^X_{[^t_i}$ 
we build a new plus-height regular expression $r.e^{(r,X)}_{[^t_i}=h^r_{{\cal G}}(r.e^{(l,X)}_{[^t_i})$, where $h^r_{{\cal G}}$ is the mirror image of 
$h_{{\cal G}}$. Consider $r.e^X_{[^t_i}=r.e^{(l,X)}_{[^t_i}r.e^{(r,X)}_{[^t_i}$. For a certain $X$ and $[^t_i$, 
denote by ${\cal P}.e^X_{[^t_i}$ the set of all (plus-height) regular expressions $r.e^X_{[^t_i}$ obtained as 
above. Furthermore,  ${\cal P}.e^X=\bigcup_{[^t_i, ]^t_i \in N^{(1)}}{\cal P}.e^X_{[^t_i}$, and 
${\cal P}.e={\cal P}.e^S\cup (\bigcup_{[_i, ]_i\in N^{(3)}}{\cal P}.e^{]_i}$).  
\vspace{-0.1cm}

Note that linear languages do not need an extended dependency graph. The set of 
all regular expressions ${\cal P}.e^S$ suffices to build a regular language in the Chomsky-Sch\"utzenberger 
theorem (see Theorem 3.5) that cannot be further adjusted by using the graphical method proposed in this section. 
Furthermore $|{\cal R}.e^S|\leq |{\cal P}.e^S|$. Equality takes place only for the case when 
each regular expression in ${\cal R}.e^S$ is a plus-height regular expression (see Example 3.6). For the case of 
context-free  languages the plus-height regular expressions in ${\cal P}.e$ must be linked with each other in such a way 
it approximates, as much as possible, the trace-language associated with the respective context-free  language. 
\vspace{0.1cm}

In order to find an optimal connection of the regular expressions in ${\cal P}.e$ we consider the following 
labeling procedure of elements in ${\cal P}.e$. Denote by $c_0$ the cardinality of ${\cal P}.e^S$, i.e., 
$|{\cal P}.e^S|=c_0$, and by $c_j$ the cardinality of ${\cal P}.e^{]_j}$, where  $[_j, ]_j \in N^{(3)}$. Each 
regular expression $r \in {\cal P}.e^S$ is labeled by a unique $q$, $1\leq q\leq c_0$, and each regular expression 
$r \in {\cal P}.e^{]_j}$, is labeled by a unique $q$, such that $\sum_{r=0}^{i-1} c_r+1\leq q\leq \sum_{r=0}^i c_r$, 
$1\leq i\leq s$, and $s=|\{]_j| [_j, ]_j \in N^{(3)}\}|$. Denote by $r^q$ the labeled version of $r$. To preserve 
symmetric structures that characterize trace-words of context-free  languages, then when we link regular expressions 
in ${\cal P}.e$ between them, each bracket in a regular expression $r^q$ is upper labeled by $q$. Exception 
makes the first bracket occurring in $r^q$ (which is a bracket in $\{]_j| [_j, ]_j \in N^{(3)}\}$). Now, a 
refined extended digraph can be built similar to that described in Construction 3.4. 
\vspace{-0.3cm}

To have a better picture of how the labeled regular expressions must be linked to each other, and 
where further relabeling procedures may be required (to obtain a better approximation of the trace-language), 
we first build for each plus-height regular expression $r^q \in {\cal P}.e^{]_j}$, $[_j, ]_j \in N^{(3)}$, a 
digraph and then we connect all digraphs between them. Denote by 
$G^{q,]_j}$ the digraph associated with $r^q \in {\cal P}.e^{]_j}$, such that $]_j$ is the initial vertex and 
the final vertex is the last bracket occurring in $r^q$. Each digraph $G^{q,]_j}$  read from the initial vertex 
$]_j$ to the final vertex provides the regular expression $r^q$. Hence, any digraph $G^{q,]_j}$ has vertices 
labeled by brackets of the forms $\{[^q_j| [_j, ]_j \hspace{-0.1cm}\in N^{(1)}\cup N^{(2)}_r\cup  N^{(3)}\} 
\cup \{]^q_j| [_j, ]_j\hspace{-0.1cm}\in N^{(2)}_l\cup N^{(2)}_r\cup  N^{(3)}\}$, $c_0\leq q \leq \sum_{r=0}^s c_r$, 
with the exception of the initial vertex $]_j$, $[_j, ]_j \in N^{(3)}$. Some of vertices in $G^{q,]_j}$, besides 
the $q$-index, may also be $\hbar$-marked, in order to prevent repetitions of the same vertex which may occur in 
a plus-height regular expression. As the construction of the dependency graph does not depend on $\hbar$-markers, 
unless it is necessary, we avoid $\hbar$-marked notations in further explanations when building this digraph.
\vspace{0.1cm}
      
The adjacent vertex $Y$ to $]_j$, in $G^{q,]_j}$, is called {\it sibling}. Any edge of the form $]^{-}_l]^{-}_k$, 
where $[_l, ]_l\in N^{(3)}$, $[_k, ]_k\in N^{(2)}_r\cup N^{(3)}$, is called {\it dummy edge}, while $]^{-}_l$ ($]^{-}_k$, if 
$[_k, ]_k\in N^{(3)}$) is a {\it dummy vertex}. An edge that is not a dummy edge is called {\it stable edge}. Denote by ${\cal G}^{]_j}$ 
the set of all digraphs $G^{q,]_j}$, i.e., their initial 
vertex is $]_j$. Any digraph $G^{q,]_j}$ has only one bracket $[^q_k$, $[_k, ]_k \in N^{(1)}$, which stands for a core segment in a trace-word. 
Right brackets $]^q_j$, $[_j, ]_j \in N^{(2)}_r\cup  N^{(3)}$, must be symmetrically arranged according to their left pairwise $[^q_j$, $[_j, ]_j \in N^{(2)}_r\cup  N^{(3)}$, that occur at the left side of $[^q_k$. A dummy vertex labeled by $]^q_j$, $[_j, ]_j \in N^{(3)}$, 
allows the connection with any digraph in ${\cal G}^{]_j}$. A digraph in ${\cal G}^{]_j}$ with a final vertex labeled by a bracket $[^{-}_k$, 
$[_k, ]_k \in N^{(1)}$, or by a bracket $]^{-}_l$, $[_l, ]_l \in N^{(2)}_r$, is called {\it terminal}, because the vertex $[^{-}_k$ 
or $]^{-}_l$, respectively, does not allow more connections.
%\footnote{When a terminal digraph $G^{q',]_j}$ is connected to ${\cal G}.e^S$ the edge $]^q_jZ$ 
%&is removed from ${\cal G}.e^S$ and a new edge $[^{q'}_kZ$ or $]^{q'}_kZ$ is added to 
%${\cal G}.e^S$ which will be permanently kept in ${\cal G}.e^S$ (since $]^{q'}_k$ or $]^{q'}_k$, does not allow more connections in 
%${\cal G}.e^S$, see the procedure described at $\bf{Step 2.}$).} to ${\cal G}.e^S$.  

Next we describe the procedure that builds a refined extended digraph with the property that reading this 
digraph (in which each loop is a plus-loop) from the initial vertex (which is $S$) to all its final vertices, 
we obtain those (plus-height) regular expressions that form a regular language that provides the best approximation 
of the corresponding trace-language. 
%\vspace{0.1cm}
   
{\bf Step 1.} First we build a digraph ${\cal G}.e^S$ that describes all (plus-height) regular expressions in ${\cal P}.e^S$. 
This can be done by connecting all digraphs in ${\cal G}^S$ to $S$.  Since each bracket labeling a vertex in $G^{q,S}$, $1\leq q \leq c_0$, is 
uniquely labeled by $q$, and there exists a finite number of brackets, ${\cal G}.e^S$ is correct (in the sense that 
it is finite and any vertex occurs only one time). The initial vertex of ${\cal G}.e^S$ is $S$. If a graph in 
${\cal G}^S$ has a final vertex labeled by a bracket $[^q_i$, $[_i,]_i \in N^{(1)}$ or by a bracket $]^q_i$, 
$[_i,]_i \in N^{(2)}_r$, then this is also a final vertex in ${\cal G}.e^S$. 

If $G_k$ is a grammar in linear-Dyck normal form then ${\cal G}.e^S$, built in this way, suffices to build the regular 
language in the Chomsky-Sch\"utzenberger theorem. The set of all paths from $S$ to each final vertex to which we apply 
the homomorphism $h_k$, defined in the proof of Theorem 3.5, yields a regular language $R_m$ that cannot be further adjusted, 
such that the Chomsky-Sch\"utzenberger theorem still holds. Therefore, we call the $R_m$ language, as minimal with respect to 
the grammar $G_k$ and the Chomsky-Sch\"utzenberger theorem, i.e., the equality $\varphi (D_K\cap R_m)=\varphi(\L(G_k))$ still 
holds, where $\varphi$ is the homomorphism defined in the proof of Theorem 2.9.    
%\vspace{0.1cm}

{\bf Step 2.} For each vertex $]^q_j$ existing in ${\cal G}.e^S$, such that $]_j\in  N^{(3)}$, we connect 
all digraphs in ${\cal G}^{]_j}$ to ${\cal G}.e^S$. This can be done by adding to ${\cal G}.e^S$ a new edge 
$]^q_jY$, for each sibling $Y$ of $]_j$ (in ${\cal G}^{]_j}$). If $Z$ is the adjacent vertex of $]^q_j$ (in the former version of 
${\cal G}.e^S$), i.e., $]^q_jZ$ is a dummy edge, then we remove in ${\cal G}.e^S$ the edge $]^q_jZ$, while 
in $G^{q',]_j}$ (connected to ${\cal G}.e^S$ through $]^q_j$) we remove the vertex $]_j$ and consequently, the 
edge $]_jY$. For the moment, all the other edges in $G^{q',]_j}$ are preserved in ${\cal G}.e^S$, too. Besides, 
if $V$ is the final vertex of $G^{q',]_j}$, then a new edge $VZ$ is added to ${\cal G}.e^S$. If $V\in \{[^{-}_k| [_k, ]_k \in N^{(1)}\}\cup \{]^{-}_l| [_l, ]_l \in N^{(2)}_r\}$, i.e., $G^{q',]_j}$ is a terminal digraph then the edge 
$VZ$ is a {\it glue edge}, i.e., it is a stable edge that makes the connection of  $G^{q',]_j}$ into ${\cal G}.e^S$ 
(or more precisely the connection of  $G^{q',]_j}$ to $G^{q,]_j}$ digraph in which it has been inserted). Otherwise,  $VZ$ is a dummy edge, which will be removed at a further connection with a digraph in ${\cal G}^{V}$. Since for 
the case of linear languages generated by a grammar in linear-Dyck normal form, ${\cal G}.e^S$ does not contain any dummy 
vertex, the construction of ${\cal G}.e^S$ is completed at {\bf {\it Step 1}}. 
\vspace{0.1cm}

A vertex in ${\cal G}.e^S$ labeled by a bracket $]^q_j$, $[_j, ]_j\in N^{(3)}$, that has no adjacent vertex, i.e., 
the out degree of the vertex labeled by $]^q_j$ is $0$, is called {\it pop vertex}. When connecting a digraph $G^{q',]_j}$ to ${\cal G}.e^S$, through a pop vertex, if $G^{q',]_j}$ is a terminal digraph, then the final vertex 
of $G^{q',]_j}$ becomes a final vertex of ${\cal G}.e^S$. If $G^{q',]_j}$ is not a terminal digraph, 
then the final vertex of $G^{q',]_j}$ becomes a pop vertex for ${\cal G}.e^S$.
\vspace{0.1cm}

%is a final vertex in  then we connect each graph $G^{q',]_j}$ to ${\cal G}.e^S$ without removing any edge in ${\cal G}.e^S$. In this 
%case, if $[^{q'}_k$ or $]^{q'}_m$ is a final vertex in $G^{q',]_j}$, such that $[_k, ]_k\in N^{(1)}$ and $[_m, ]_m)\in N^{(2)}_r$, then $[^{q'}_k$ 
%or $]^{q'}_m$ is also a final vertex in ${\cal G}.e^S$, respectively. For any final vertex in  

If there exist more then one vertex labeled by an upper indexed\footnote{As ${\cal G}.e^S$ is finite, there cannot exist in ${\cal G}.e^S$ two right brackets $]_j$, $[_j, ]_j\in  N^{(3)}$, upper indexed by the same value.} bracket $]^{\bar q}_j$, $[_j, ]_j\in  N^{(3)}$, then, if $G^{q',]_j}$ has been already added to  ${\cal G}.e^S$ there is 
no need to add another ``copy'' of $G^{q',]_j}$. It is enough to connect $]^{\bar q}_j$ to the digraph existing 
in ${\cal G}.e^S$, i.e., to add a new edge $]^{\bar q}_jY$, where $Y$ is a sibling of $]_j$ in $G^{q',]_j}$. 
This observation holds for any element in ${\cal G}^{]_j}$. 
%If the added digraph $G^{q',]_j}$ contains a label $]^{q'}_j$, $[_j, ]_j\in N^{(3)}$, then instead of adding a new ``copy'' of $G^{q',]_j}$ 
%(or of any other digraph in ${\cal G}^{]_j}$) a new edge $]^{q'}_jY$ ($Y$ is a sibling of $]_j$ in $G^{q',]_j}$) is added to ${\cal G}.e^S$ 
%(this is to make the connection of $]^{q'}_j$ to $G^{q',]_j}$, already added to ${\cal G}.e^S$). 
The procedure described at {\bf {\it Step 2}} is repeated for each new dummy or pop vertex added to ${\cal G}.e^S$. For each transformation performed on ${\cal G}.e^S$, we maintain the same notation ${\cal G}.e^S$ for the new 
obtained digraph. The construction of ${\cal G}.e^S$ ends up then when each vertex $]^{-}_j$, $[_j, ]_j\in N^{(3)}$, has been connected to a digraph in ${\cal G}^{]_j}$, i.e., no dummy and pop vertices exist in ${\cal G}.e^S$. The 
only permissible contexts under which a bracket $]^{-}_j$, $[_j, ]_j\in N^{(3)}$, may occur, in the final version 
of ${\cal G}.e^S$, are of the forms $]^{-}_i]^{-}_j [^{-}_k$, $[^{-}_h]^{-}_j [^{-}_k$, $]^{-}_i]^{-}_j ]^{-}_l$, $[^{-}_h]^{-}_j ]^{-}_l$, where $[_i, ]_i\in N^{(2)}_r$, $[_h, ]_h\in N^{(1)}$, $[_k, ]_k\in N^{(1)}\cup N^{(2)}_r
\cup N^{(3)}$, $[_l, ]_l\in N^{(2)}_l$.  
\vspace{0.1cm}

There are several refinements that can be done on ${\cal G}.e^S$ such that the resulted regular language better approximates the trace  
language associated with the considered context-free  language. Two peculiar situations may occur when adding digraphs to 
${\cal G}.e^S$: 
\vspace{0.1cm}

$\bf{{\cal I}_1}$. First, suppose that during the construction of ${\cal G}.e^S$ 
by subsequently connecting digraphs between them, starting from $]^q_j$,  $[_j, ]_j\in N^{(3)}$, we reach a terminal digraph with 
a final vertex $]^{q'}_k$, $[_k, ]_k\in N^{(2)}_r$, such that $]^{q'}_k$ is linked to $Z$, forming thus a stable (glue) edge $]^{q'}_kZ$. 
Denote by $\wp=]^q_j...]^{q'}_kZ$ the path (in ${\cal G}.e^S$) from $]^q_j$ to $]^{q'}_kZ$, obtained at this stage. If the vertex that 
precedes $]^{q'}_k$ in $\wp$ is $]^{q'}_j$, $[_j, ]_j\in N^{(3)}$, i.e., $\wp= ]^q_j... ]^{q'}_j]^{q'}_kZ$, then connecting $]^{q'}_j$ 
($]^{q'}_j]^{q'}_k$ is a dummy edge), through its siblings, to digraphs in ${\cal G}^{]_j}$ another edge $]^{q'}_k]^{q'}_k$ preceded 
by $]^{q'}_j]^{q'}_k$, is added to ${\cal G}.e^S$, i.e., $\wp$ becomes $\wp= ]^q_j...]^{q'}_j(]^{q'}_k)^2Z$. Since $]^{q'}_j]^{q'}_k$ 
is a dummy edge, the vertex $]^{q'}_j$ must be again connected to digraphs in ${\cal G}^{]_j}$, and so on, until $]^{q'}_j$ is connected 
to a terminal digraph $G^{\bar q,]_j}\in {\cal G}^{]_j}$, $\bar q\neq q'$, that has a final vertex labeled by a bracket 
$]^{\bar q}_m$, $[_m, ]_m\in N^{(2)}_r$ ($m$ and $k$ not necessarily distinct), or by a bracket $[^{\bar q}_m$, $[_m, ]_m\in N^{(1)}$
such that $]^{\bar q}_m$ is not preceded by a bracket of the form $]^{-}_j$, $[_j, ]_j\in N^{(3)}$. Then $\wp$ will be either of the 
form $]^q_j\wp_1]^{\bar q}_m(]^{q'}_k)^+Z$ or of the form $]^q_j\wp_1[^{\bar q}_m(]^{q'}_k)^+Z$, respectively. On the other hand, since 
$G^{\bar q,]_j}\in {\cal G}^{]_j}$ the digraph $G^{\bar q,]_j}$ can be added to ${\cal G}.e^S$, through $]^q_j$, from the very first 
beginning, avoiding thus the plus-loop $(]^{q'}_k)^+$, i.e., there should exist in ${\cal G}.e^S$ a new path $\wp'=]^q_j\wp_2]^{\bar q}_mZ$ 
or $\wp'=]^q_j\wp_2[^{\bar q}_mZ$ (where $\wp_2$ is a path in $G^{\bar q,]_j}$). This allows two other new paths to be created in  
${\cal G}.e^S$, i.e., $\bar{\wp}=]^q_j\wp_2]^{\bar q}_m(]^{q'}_k)^+Z$ (or $\wp''=]^q_j\wp_2[^{\bar q}_m(]^{q'}_k)^+Z$) 
and  $\bar{\wp}'=]^q_j\wp_1]^{\bar q}_mZ$ (or $\bar{\wp}''=]^q_j\wp_1[^{\bar q}_mZ$), which are of no use in approximating the trace 
language (hence in building the regular language in the Chomsky-Sch\"utzenberger theorem). Paths $\bar{\wp}$ and $\bar{\wp}'$ ($\wp''$, $\bar{\wp}''$) do not 
affect the intersection with the Dyck language but they enlarge the regular language with useless words. 

In order to avoid the paths $\bar{\wp}$ and $\bar{\wp}'$ (or $\wp''$, $\bar{\wp}''$) the terminal digraph $G^{\bar q,]_j}$ receives a 
new label $\tilde q$, besides of label $\bar q$ (which is maintained to allow $\wp$ to be produced). To allow the shorter path $\wp'$ to be 
created, instead of $G^{\bar q,]_j}$ the terminal digraph $G^{\tilde q,]_j}$ is connected to ${\cal G}.e^S$ through the dummy vertex $]^q_j$. 
Hence $\wp'$ becomes $\wp'\hspace{-0.1cm}=]^q_j...]^{\tilde q}_mZ$ (or $\wp'\hspace{-0.1cm}=]^q_j...[^{\tilde q}_mZ$) , while $\wp$ remains 
$]^q_j...]^{\bar q}_m(]^{q'}_k)^+Z$ (or $]^q_j...[^{\bar q}_m(]^{q'}_k)^+Z$, respectively). This relabeling procedure is used for any case 
similar to that described above\footnote{For instance, $]^{q'}_k$ may also be a dummy vertex and $]^{q'}_kZ$ a dummy edge.} encountered 
during the computation of ${\cal G}.e^S$. As there may exist a finite number\footnote{The plus-height of a regular expression obtained from any 
digraph in ${\cal G}^{]_j}$ is finite related to the length of the strings in $L(G_k)$.} of plus-loops in ${\cal G}.e^S$, there will be 
a finite number of ``relabeled'' digraphs (not necessarily terminal). A loop (not necessarily a self-loop) may be reached through different paths 
that must be ``renamed'' (if we want to avoid that loop). 

%Another peculiar situation occur when $Z$ is equal to $]^{q'}_k$ and  $]^{q'}_k$ is not preceded by $]^{q'}_j$ (in $G^{q',]_j}$)

%To be also observed that the ``mirror image'' of loops in $]^{q'}_k$, in a trace-word, is a ``broken loop'' (it is a loop because it is yielded at 
%each repetition of the digraph that ends in $]^{q'}_k$, it is broken because ...) having the same plus-height of $[^{q'}_k$. 
%If $[^{q'}_k$ occur only one time, then $]^{q'}_k$ must also occur of only one time.   
\vspace{0.1cm}

$\bf{{\cal I}_2}$. Another situation that requires a relabeling procedure may occur when connecting a digraph to ${\cal G}.e^S$ through a pop vertex. Suppose that $]^q_j$, $[_j, ]_j\in N^{(3)}$, is a pop vertex, and the 
digraph $G^{q',]_j}$ that must be added to ${\cal G}.e^S$ has been already connected through a dummy vertex labeled 
by $]^{\bar q}_j$ (i.e., $G^{q',]_j}$ has been already inserted in ${\cal G}.e^S$). According to the procedure described at {\bf {\it Step 2}} the vertex $]^q_j$ is linked to the sibling of $]_j$ in $G^{q',]_j}$ already existing 
in ${\cal G}.e^S$. Since the connection of $G^{q',]_j}$ to ${\cal G}.e^S$ has been done through a dummy vertex, 
the final vertex in $G^{q',]_j}$ cannot be neither a final vertex in ${\cal G}.e^S$ (if  $G^{q',]_j}$ is a terminal digraph) nor a pop vertex.
%unless we allow a pop vertex to be also a dummy vertex. 
\vspace{0.1cm}

To forbid a pop vertex $]^{-}_j$ to overlap with a dummy vertex $]^{-}_j$, each of the digraphs connected to ${\cal G}.e^S$ through a pop vertex, 
is renamed by a new label. Denote by $\bar {\cal G}^{]_j}$ the labeled version of ${\cal G}^{]_j}$. Then connections through pop vertices will 
be done by using only digraphs in $\bar {\cal G}^{]_j}$. However, any dummy vertex $]^{-}_j$, that is not a pop vertex, obtained by connecting 
digraphs in $\bar {\cal G}^{]_j}$ to ${\cal G}.e^S$ should be connected to the original digraphs in ${\cal G}^{]_j}$, unless a relabeling procedure 
described at $\bf{{\cal I}_1}$ is required. 
%There are situation when the relabeling procedure described at $\bf{{\cal I}_2}$ can be avoided (because applying it leads to the same results). 
%It is the case when there exists in ${\cal G}^S$ a digraph that ends in a dummy vertex  $]^{-}_j$  which forms a self-loop. 
\vspace{0.1cm}

Denote by $\bar N_k=\{[^{-}_i| [_i, ]_i \in N^{(1)} \cup N^{(2)}_r \cup N^{(3)}\}\cup \{]^{-}_j|[_j, ]_j \in N^{(2)}_l\cup N^{(2)}_r\cup N^{(3)}\}$, 
the set of vertices composing ${\cal G}.e^S$, in which some brackets may be $\hbar$-marked (by distinct $\hbar$-markers). To reach the regular 
language in the Chomsky-Sch\"utzenberger theorem we denote by ${\cal R}_{{\cal G}}$ the set of all regular expressions obtained by reading  
${\cal G}.e^S$ from the initial vertex $S$ to any final vertex. First, suppose that $G_k$ does not have an extended grammar. We have $K=k$ and 
$D'_k=\L(G_k)$. Consider the homomorphism $h_k\hspace{-0.1cm}: \bar N_k \cup \{S\}\hspace{-0.1cm}\rightarrow \{[_i,]_i| [_i,]_i \hspace{-0.1cm}
\in N^{(2)}_r\cup N^{(3)}\}\cup\{[_i]_i| [_i,]_i \hspace{-0.1cm}\in N^{(2)}_l\cup N^{(1)}\}\cup \{\lambda\}$, defined by $h_k(S)=\lambda$, 
$h_k([^{-}_i)= [_i$, $h_k(]^{-}_i)=\hspace{0.1cm}]_i$ for any $[_i, ]_i \in N^{(2)}_r$, $h_k(]^{-}_i)=[_i\hspace{0.1cm}]_i$ for any 
$[_i, ]_i \in N^{(2)}_l$, $h_k([^{-}_i)=[_i\hspace{0.1cm}]_i$ for any $[_i, ]_i \in N^{(1)}$. Then $R_m=h_k({\cal R}_{{\cal G}})$ is a 
regular language with  $D_k\cap R_m=\L(G_k)$. Furthermore, $R_m$ is a strength refinement of $R$, such that the 
Chomsky-Sch\"utzenberger theorem still holds. This is because when building regular expressions in ${\cal P}.e$ each 
$r.e^{(l,X)}_{[^t_i}$ is linked only to its right pairwise $r.e^{(r,X)}_{[^t_i}$ (due to plus-height 
considerations and labeling procedures). In this way all plus-loops in $r.e^{(l,X)}_{[^t_i}$ are correctly 
mirrored (through $h^r_{\cal G}$) into its correct pairwise $r.e^{(r,X)}_{[^t_i}$. The case of $\lambda$-loops 
is taken by the relabeling procedure described at $\bf{{\cal I}_1}$. This is also applicable each time we want 
to fork a path in ${\cal G}.e^S$ in order to avoid useless loops on that path. The relabeling procedure 
$\bf{{\cal I}_2}$ allows to leave ${\cal G}.e^S$ without re-loading another useless path. 
That is why the regular language $R_m$ built this way is a tighter approximation of $\L(G_k)$. A finer language 
than $R_m$ can be found by searching for a more efficient grammar in Dyck normal form, with respect to 
the number of rules and nonterminals.   
\vspace{0.1cm}

If $G_k$ has an extended grammar $G_{k+p}=(N_{k+p}, T, P_{k+p}, S)$ (built as in the proof of Theorem 2.9) then 
${\cal R}_{{\cal G}}$ is augmented with ${\cal r}_e=\{S[_{t_{k+1}}, ..., S[_{t_{k+p}}\}$ and $h_k$ is extended to 
$h_K$, $h_K\hspace{-0.1cm}: \bar N_k \cup \{S\}\cup \{[_{t_{k+1}}, ..., [_{t_{k+p}}\} \rightarrow 
\{[_i,]_i| [_i,]_i \in N^{(2)}_r\cup N^{(3)}\}\cup \{[_i]_i| [_i,]_i \in N^{(2)}_l\cup N^{(1)}\} 
\cup \{[_{t_{k+1}}]_{t_{k+1}}, ..., [_{t_{k+p}}]_{t_{k+p}}\}\cup \{\lambda\}$,  $h_K(x)=h_k(x)$, $x\notin \{[_{t_{k+1}}, ..., [_{t_{k+p}}\}$, and $h_K([_{t_{k+j}})=[_{t_{k+j}}]_{t_{k+j}}$, $1\leq j\leq p$, $K=k+p$. 

\vspace{0.3cm}
\begin{figure}[h!]
	\centering
		\includegraphics[scale=0.70]{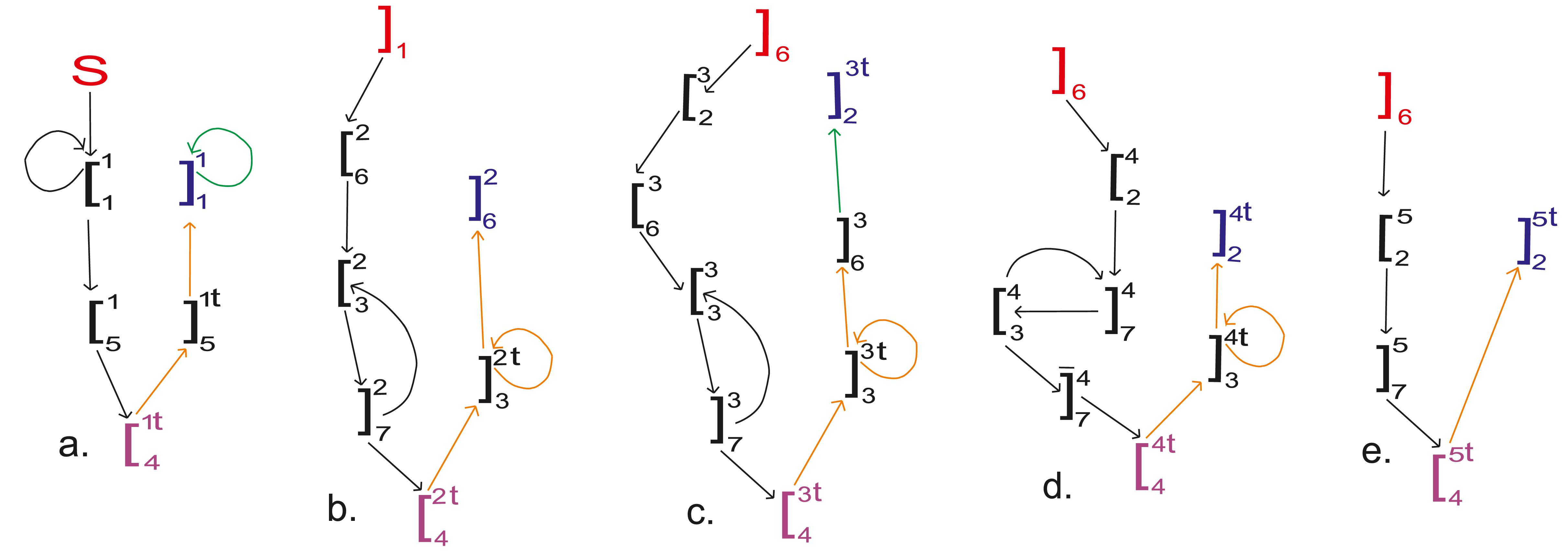}
		\vspace{-0.1cm}
	\caption{{\footnotesize a. - e. Graphs associated with regular expressions in ${\cal P}.e$ (Example 4.2).  
     Initial vertices are colored in red, final vertices in blue, while purple vertices mark a 
     core segment. $\bar ]^4_7$ is a marked vertex to allow the plus-loop $([^4_3]^4_7)^+$.}} 
\end{figure}
 \vspace*{-0.4cm}
 
\begin{example}\rm
Consider the context-free  grammar in Example 3.7 with the dependency graphs sketched in Figure 3. 
The set ${\cal P}.e$ of labeled plus-loop regular expressions built from\\
the dependency graphs is composed of $S([^1_1)^+[^1_5[^{1t}_4]^{1t}_5(]^1_1)^+$ (with the associated digraph 
$G^{1,S}$, Fig. 3.a), $]_1[^2_6([^2_3]^2_7)^+[^{2t}_4(]^{2t}_3)^+]^2_6$ (with $G^{2,]_1}$,  Fig. 3.b),
$]_6[^3_2[^3_6([^3_3]^3_7)^+[^{3t}_4(]^{3t}_3)^+]^3_6]^{3t}_2$ (with $G^{3,]_6}$,  Fig. 3.c), 
$]_6[^4_2(]^4_7[^4_3)^+]^4_7[^{4t}_4(]^{4t}_3)^+]^{4t}_2$ or the $\hbar$-marked version 
$]_6[^4_2(]^4_7[^4_3)^+\bar ]^4_7[^{4t}_4(]^{4t}_3)^+]^{4t}_2$ (with the associated digraph $G^{4,]_6}$,  Fig. 3.d), 
and $]_6[^5_2]^5_7[^{5t}_4]^{t5}_2$ (with the digraph $G^{5,]_6}$,  Fig. 3.e). 
\vspace*{0.2cm}

The extended dependency graph built with respect to the refinement procedure is sketched in Figure 4. 
The terminal digraphs $G^{6,]_6}$ and $G^{7,]_6}$ are introduced with respect to the relabeling procedure 
$\bf{{\cal I}_1}$, in order to prevent the loop yielded by the ''iterated`` digraph $G^{3,]_6}$ to occur 
between $G^{2,]_1}$ and $G^{6,]_6}$ (or $G^{7,]_6}$). It also forbids the self-loop $(]^{3t}_2)^+$ to be 
linked to $G^{6,]_6}$ (or to $G^{7,]_6}$), then when the digraph $G^{3,]_6}$ is not added to the corresponding 
path. Due to the self-loop $(]^1_1)^+$, in which $]^1_1$ is a pop vertex, we did not applied the relabeling 
procedure described at $\bf{{\cal I}_2}$ (applying it leads to the same result).  
\end{example}

\begin{figure}
	\centering
		\includegraphics[scale=0.69]{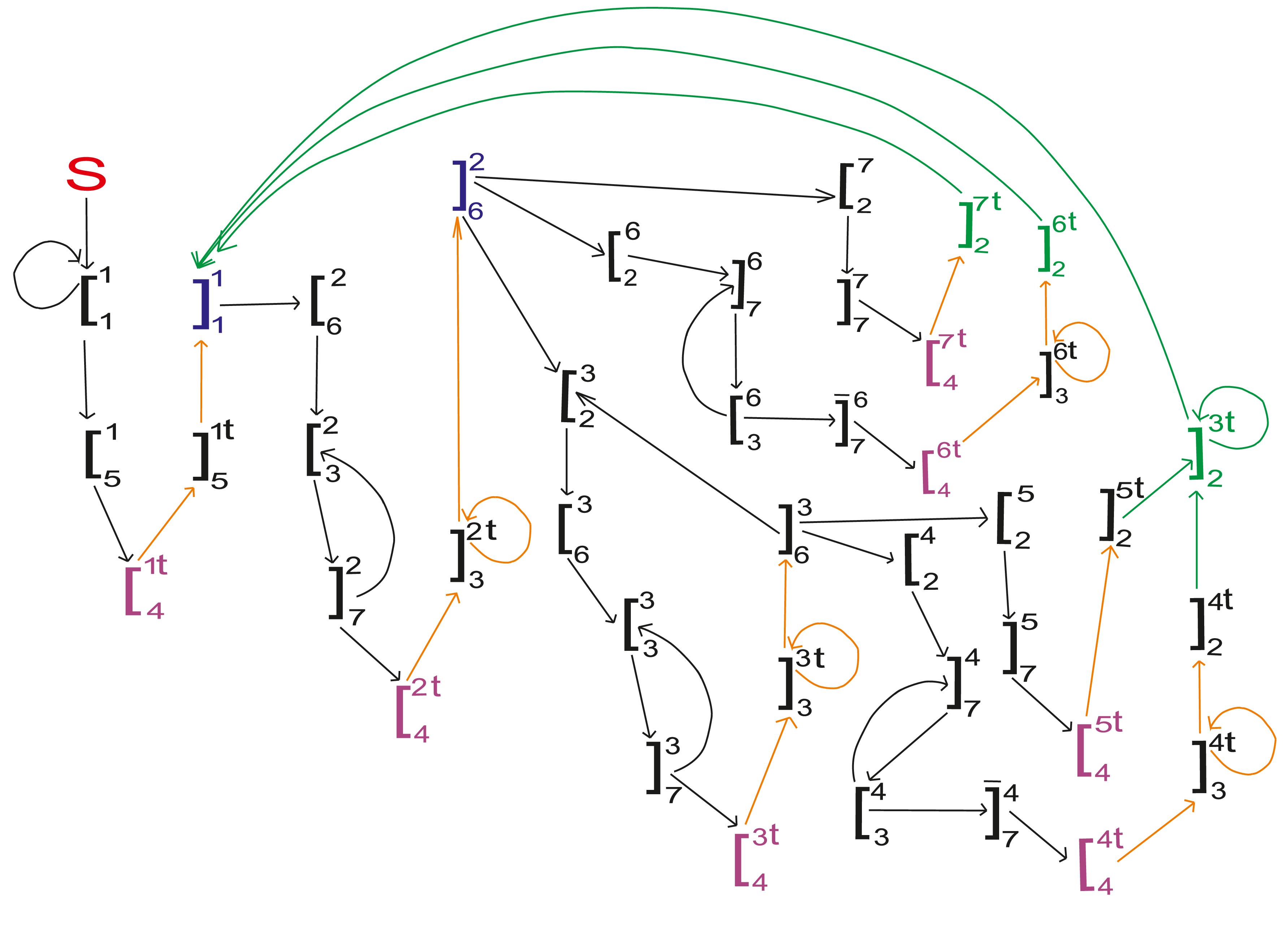}
		\vspace{-0.4cm}
	\caption{{\footnotesize The refined dependency graph of the context-free  grammar in Examples 3.7 and 4.2. $S$ is the 
	initial vertex, vertices colored in green are final vertices, vertices colored in blue are dummy vertices, 
	vertices colored in purple mark a core segment. Orange edges emphasize symetrical structures built with respect 
	to the structure of the trace language. Green edges are glue edges.}} 
\end{figure}
\vspace*{-0.2cm}

\section{A Regular Superset Approximation for Context-Free\\Languages}

A regular language $R$ may be considered a superset approximation for a context-free  language $L$, if $L\subseteq R$. A 
good approximation for $L$ is that for which the set $R-L$ is as small as possible. There are considerable methods 
to find a regular approximation for a context-free  language. The most significant consist in building, through several 
transformations applied to the original pushdown automaton (or context-free   grammar), the most appropriate finite automaton 
(regular grammar) recognizing (generating) a regular superset approximation of the original context-free  language.  How 
accurate the approximation is, depends on the transformations applied to the considered devices. However, the perfect regular 
superset (or subset) approximation for an arbitrary context-free language cannot be built. For surveys on approximation methods 
and their practical applications in computational linguistics (especially in parsing theory) the reader is referred to \cite{MN} 
and \cite{N}. Methods to measure the accuracy of a regular approximation can be found in \cite{CS}, \cite{ER}, and \cite{SB}. 
\vspace*{0.1cm}

In the sequel we propose a new approximation technique that emerges from the Chomsky-Sch\"utzenberger theorem. In brief, the method 
consists in transforming the original context-free  grammar into a context-free  grammar in Dyck normal form. For this grammar we build 
the refined extended dependency graph ${\cal G}.e^S$ described in Section 4. From ${\cal G}.e^S$ we depict a {\it state diagram} ${\cal A}_e$ 
for a finite automaton and a regular grammar $G_r=(N_r, T, P_r, S)$ that generates a regular (superset) approximation for $L(G_k)$ (which is nothing else than the 
image through $\varphi$ of the language $R_m$ built in Section 4). 

Let $G_k=(N_k, T, P_k, S)$ be an arbitrary context-free  grammar in Dyck normal form, and ${\cal G}.e^S=(V_e, E_e)$ the extended dependency 
graph of $G_k$. Recall that $V_e=\{[^{-}_i| [_i, ]_i \in N^{(1)} \cup N^{(2)}_r \cup N^{(3)}\}
\cup \{]^{-}_j|[_j, ]_j \in N^{(2)}_l\cup N^{(2)}_r\cup N^{(3)}\} \cup \{S\}$ in which some of the vertices may 
be $\hbar$-marked, in order to prevent repetition of the same bracket when building the digraph associated with a plus-height regular 
expression. In brief, the state diagram ${\cal A}_e$ can be built by skipping  in ${\cal G}.e^S$ all left brackets in $N^{(2)}_r$ and 
all brackets in $N^{(3)}$, and labeling the edges with the symbol produced by left or right bracket in $N^{(2)}\cup N^{(1)}$. This reasoning 
is applied no matter whether the vertex in $V_e$ is $\hbar$-marked or not. Therefore, we avoid $\hbar$-marker specifications when building 
${\cal A}_e$, unless this is strictly necessary. Denote by $s_f$ the accepting state of ${\cal A}_e$. The {\it start state} of ${\cal A}_e$ 
is $s_S$, where $S$ is the axiom of $G_k$.  We proceed as follows:

$1.$ There exists an edge in ${\cal A}_e$ from $s_S$ to $s_{]^q_i}$, labeled by $a$, where $[_i, ]_i\in N^{(2)}_l$ and 
$[_i\rightarrow a\in P_k$, if either $S]^q_i \in E_e$ or there exists a path in ${\cal G}.e^S$ from $S$ to $]^q_i$ that 
contains no vertex labeled by $]^q_j$, $[_j, ]_j\in N^{(2)}_l$, or by $[^{-t}_k$, $[_k, ]_k\in N^{(1)}$. We fix 
$S\rightarrow a]^q_i\in P_r$. 
\vspace{0.1cm}

$2.$ There exists an edge in ${\cal A}_e$ from $s_S$ to $s_{[^{qt}_i}$, labeled by $a$, and an edge 
from $s_{[^{qt}_i}$ to $s_{]^{qt}_i}$ labeled by $b$, where $[^t_i, ]^t_i\in N^{(1)}$, $[^t_i\rightarrow a$, and 
$]^t_i\rightarrow b \in P_k$, if either $S[^{qt}_i \in E_e$ or there exists a path in ${\cal G}.e^S$ 
from $S$ to $[^{qt}_i$ that contains no vertex labeled by $]^q_j$, $[_j, ]_j\in N^{(2)}_l$, or labeled by 
$[^{-t}_k$, $[_k, ]_k\in N^{(1)}$. We fix $S\rightarrow a[^{qt}_i, [^{qt}_i \rightarrow b]^{qt}_i\in P_r$. 
\vspace{0.1cm}

$3.$ There exists an edge in ${\cal A}_e$ from $s_{]^q_i}$ to $s_{]^q_j}$, labeled by $a$, where $[_i, ]_i, [_j, ]_j \in N^{(2)}_l$ 
and $[_j\rightarrow a\in P_k$, if either $]^q_i]^q_j \in E_e$ or there exists a path in ${\cal G}.e^S$ from 
$]^q_i$ to $]^q_j$ that contains no vertex labeled by $[^{qt}_k$ or by $]^q_l$, $[_k, ]_k\in N^{(1)}$, $[_l, ]_l\in N^{(2)}_l$. 
If $i=j$, i.e., $]^q_i]^q_i$ is a self-loop in ${\cal G}.e^S$, then $s_{]^q_i}s_{]^q_i}$ is a self-loop\footnote{This case deals 
also with the situation when $]^q_i$, $[_i, ]_i \in N^{(2)}_l$, occurs in a loop in ${\cal G}.e^S$ composed of only 
left brackets in $N^{(2)}_r\cup N^{(3)}$, excepting $]^q_i$. A loop composed of only left brackets in 
$N^{(2)}_r\cup N^{(3)}$ is ignored when building ${\cal A}_e$.} in ${\cal A}_e$. We fix $]^q_i\rightarrow a]^q_j\in P_r$. 
\vspace{0.1cm}

$4.$ There exists an edge in ${\cal A}_e$ from $s_{]^q_i}$ to $s_{[^{qt}_j}$, labeled by $a$ and an edge 
from $s_{[^{qt}_j}$ to $s_{]^{qt}_j}$ labeled by $b$, where $[_i, ]_i\in N^{(2)}_l$, $[_j, ]_j\in N^{(1)}$,  
$[^t_j\rightarrow a$, and $]^t_j\rightarrow b \in P_k$, if either $]^q_i[^{qt}_j \in E_e$ or there exists a 
path in ${\cal G}.e^S$ from $]^q_i$ to $[^{qt}_j$ that contains no vertex labeled by $]^q_k$, $[_k, ]_k\in N^{(2)}_l$. 
We fix $]^q_i\rightarrow a[^{qt}_j, [^{qt}_j \rightarrow b]^{qt}_j\in P_r$. 
\vspace{0.1cm}

$5.$ There exists an edge in ${\cal A}_e$ from $s_{]^q_i}$ to $s_{]^{q'}_j}$, labeled by $a$, where $[_i, ]_i, [_j, ]_j\in N^{(2)}_r$, 
and $]_j\rightarrow a\in P_k$, if $]^q_i]^{q'}_j \in E_e$. If $i=j$ and $q=q'$, then $s_{]^q_i}s_{]^q_i}$ is a self-loop in ${\cal A}_e$
(because $]^q_i]^q_i$ is a self-loop in ${\cal G}.e^S$). We fix $]^q_i\rightarrow a]^{q'}_j\in P_r$. Note that, it is also possible to 
have $i\neq j$ and $q=q'$ or $q\neq q'$ (with $i=j$ or $i\neq j$, case in which $]^q_{-}]^{q'}_{-}$ is a glue edge in ${\cal G}.e^S$). 
\vspace{0.1cm}

$6.$ There exists an edge in ${\cal A}_e$ from $s_{]^q_i}$ to $s_{]^{q'}_j}$, labeled by $a$, where $[_i, ]_i\in N^{(2)}_r$, 
$[_j, ]_j\in N^{(2)}_l$, and $[_j\rightarrow a\in P_k$, if there exists a path in ${\cal G}.e^S$ from $]^q_i$ to $]^{q'}_j$ that 
contains no vertex labeled by $]^{-}_k$, $[_k, ]_k\in N^{(2)}_l\cup N^{(2)}_r$, or labeled by $[^{-t}_l$, $([_l, ]_l)\in N^{(1)}$. 
We fix $]^q_i\rightarrow a]^{q'}_j\in P_r$. Note that, $q$ may be equal to $q'$.
\vspace{0.1cm}

$7.$ There exists an edge $s_{]^q_i}s_{[^{q't}_j}$ labeled by $a$, and an edge $s_{[^{q't}_j}s_{]^{q't}_j}$ labeled by $b$, where
$[_i, ]_i\in N^{(2)}_r$, $[_j, ]_j\in N^{(1)}$, $[^t_j\rightarrow a$, and  $]^t_j\rightarrow b \in P_k$, if there exists a path in 
${\cal G}.e^S$ from $]^q_i$ to $[^{q't}_j$ that contains no vertex labeled by $]^{-}_k$, $[_k, ]_k\in N^{(2)}_l\cup N^{(2)}_r$, or 
by $[^{-t}_l$, $[_l, ]_l\in N^{(1)}$. We fix $]^q_i\hspace{-0.1cm}\rightarrow a[^{q't}_j, [^{q't}_j \rightarrow \hspace{-0.1cm}b]^{q't}_j\in P_r$. 
\vspace{0.1cm}

$8.$ There exists an edge in ${\cal A}_e$ from $s_{]^{qt}_i}$ to $s_{]^{q'}_j}$, labeled by $a$, where 
$[_i, ]_i\in N^{(1)}$, $[_j, ]_j\in N^{(2)}_r$, and $]_j\rightarrow a\in P_k$, if $[^{qt}_i]^{q'}_j \in E_e$. 
We fix $]^{qt}_i\rightarrow a]^q_j\in P_r$. Note that, it is possible to have $q=q'$ or $q\neq q'$ (in the last 
case $]^{qt}_i]^{q'}_j$ is a glue edge in ${\cal G}.e^S$). 
\vspace{0.1cm}

$9.$ There exists an edge $s_{]^{qt}_i}s_{[^{q't}_j}$, labeled by $a$, and an edge $s_{[^{q't}_j}s_{]^{q't}_j}$ 
labeled by $b$, where $[_i, ]_i, [_j, ]_j\hspace{-0.1cm}\in N^{(1)}$,  
\vspace{-0.1cm}$[^t_j\rightarrow a$, and  $]^t_j\hspace{-0.1cm}\rightarrow b \hspace{-0.1cm}\in P_k$, 
if there exists a path in ${\cal G}.e^S$ from $[^{qt}_i$ to $[^{q't}_j$ that contains no vertex labeled by $]^{-}_k$, 
$[_k, ]_k\hspace{-0.1cm}\in N^{(2)}_l\cup N^{(2)}_r$, or by $[^{-t}_l$, $[_l, ]_l\in N^{(1)}$.
We fix $]^{qt}_i\hspace{-0.1cm}\rightarrow a[^{q't}_j, [^{q't}_j \rightarrow \hspace{-0.1cm}b]^{q't}_j\in P_r$.
Note that, $[^{qt}_i$ may be equal to $[^{q't}_j$, i.e., $i=j$ and $q=q'$, i.e., the case of a loop in $[^{qt}_i$.
\vspace{0.1cm}

\hspace{-0.1cm}$10.$ - For any final vertex labeled by $]^q_i$, $[_i, ]_i\in N^{(2)}_r$, or by $[^{qt}_i$, 
$[^t_i, ]^t_i\in N^{(1)}$,  in ${\cal G}.e^S$, we add in ${\cal A}_e$ a new edge $s_{]^q_i}s_f$, or 
$s_{]^{qt}_i}s_f$, respectively. In both cases, this is labeled by $\lambda$. We set in $P_r$ a rule of the 
form $]^q_i\rightarrow \lambda$ or $]^{qt}_i\rightarrow \lambda$, respectively. 
\vspace{0.2cm}

The new grammar $G_r=(N_r, T, P_r, S)$, in which the set of rules $P_r$ is built as above, and 
$N_r=\{]^{-}_i| [_i, ]_i\in N^{(2)}\}\cup \{[^{-}_i, ]^{-}_i| [_i, ]_i\in N^{(1)}\}$ is a regular grammar 
generating a regular superset approximation for $L(G_k)$. Recall that, some of the brackets in $N_r$ may also 
be $\hbar$-marked (by distinct symbols). It is easy to observe that $L(G_r)=\varphi(R_m)$, where $\varphi$ is 
the homomorphism in the proof of Theorem 2.9. 
%which is minimal only with respect to the ``shape'' of grammar $G_k$. 
\vspace{0.1cm}

Note that since the regular language in the Chomsky-Sch\"utzenberger theorem is an approximation of the trace-language, $R_m$ depends 
on the considered context-free  grammar in Dyck normal form. Hence, the refinement of the regular approximation depicted in this section 
is considered with respect to the structure of the grammar $G_k$ in Dyck normal form, where by the {\it structure} we mean the number of 
rules and nonterminals composing $G_k$.  As for $L=L(G_k)$ there exist infinitely many grammars generating it, setting these grammars in 
Dyck normal form other trace languages can be drawn, and consequently other regular languages, of type $R_m$, can be built. The best 
approximation for $L$ is the regular language with fewer words that are not in $L$. 

Denote by ${\cal G}_L$ the infinite set of grammars in Dyck normal form generating $L$, by ${\cal R}_m$ the set of all regular languages 
obtained from the refined extended dependency graphs associated with grammars in ${\cal G}_L$, and by 
${\cal A}_L=\{ \varphi(R_m)| R_m\in {\cal R}_m\}$ the set of all superset regular approximations of $L$. It is easy 
to observe that ${\cal A}_L$, with the inclusion relation on sets, is a partially ordered subset of context-free languages. 
${\cal A}_L$ has an {\it infimum} equal to the context-free  language it approximates, but it does not have the {\it least element}. 
Indeed, as proved in \cite{Ce}, \cite{GK1}, \cite{GK2}, and \cite{GK3}, there is no algorithm to build for a certain context-free  
language $L$, the simplest context-free  grammar that generates $L$. Hence, there is no possibility to identify the simplest context-free  
grammar in Dyck normal form that generates $L$. Therefore, there is no algorithm to build the minimal superset approximation for $L$. Where 
by the {\it simplest} grammar we refer to a grammar with a minimal number of nonterminals, 
rules, or loops (grammatical levels encountered during derivations). Consequently,  ${\cal A}_L$ does not have the least element.      
%plus-height loops that characterizes derivations in that grammar (hence the trace-language or the extended dependency digraph).

It would be interesting to further study how the (refined) extended dependency graphs, associated with grammars in 
Dyck normal form generating a certain context-free  language $L$, vary depending on the structure of these grammars\footnote{For instance, 
how does it look the extended dependency graph associated with a {\it nonself-embedding} grammar in Dyck normal form, and which 
is the corresponding regular superset approximation. Note that, a context-free -nonself-embedding grammar always generates a regular 
language (since the language is finite).}, and what makes the structure of the regular language $R_m$ (hence the regular 
superset approximation) simpler. In other words, to find a hierarchy on ${\cal A}_L$, depending on the structure of the grammars in 
Dyck normal form that generate $L$. These may also provide an appropriate measure to compare languages in ${\cal A}_L$. On the other 
hand, for an ambiguous grammar $G_k$, there exist several paths (hence regular expressions) in the refined extended dependency graph, 
which ``approximate'' the same word in $L(G_k)$. Apparently, finding an unambiguous grammar for $L(G_k)$ may refine the language $R_m$. The main disadvantage is that, again in general, there is no algorithm to solve this problem. Moreover, even if it is possible to find an unambiguous grammar for $L(G_k)$, it is doubtful that the corresponding regular language $R_m$ is finer than the others. In \cite{GK1} it is also proved that 
the cost of the ``simplicity'' is the ambiguity. In other words, finding an unambiguous grammar for $L=L(G_k)$ may lead to the increase in size (e.g. number of nonterminals, rules, levels, etc.) of the respective grammar. Which again, may enlarge $R_m$ with useless words. Therefore, a challenging matter that deserves further attention is whether the unambiguity is more powerful than the ``simplicity'' in determining a more refined regular superset approximation for 
a certain context-free  language (with respect to the method proposed in this paper). 
          
%On the other hand when we have defined the trace language, hence whnen we built the refined extended dependency graph, of a 
%grammar $G_k$ in Dyck normal form we did not make any restriction for ambiguous grammars. 

In \cite{CS} it is proved that optimal (minimal) superset approximations exist for several kind of context-free  languages, 
but no specification is provided of how the existing minimal approximation can be built starting from the context-free  language it approximates.  
It would be challenging to further investigate whether there exist subsets of context-free  languages for which it would be possible to build a 
minimal superset approximation (by using the graphical method herein proposed).

\begin{figure}
	\centering
		\includegraphics[scale=0.69]{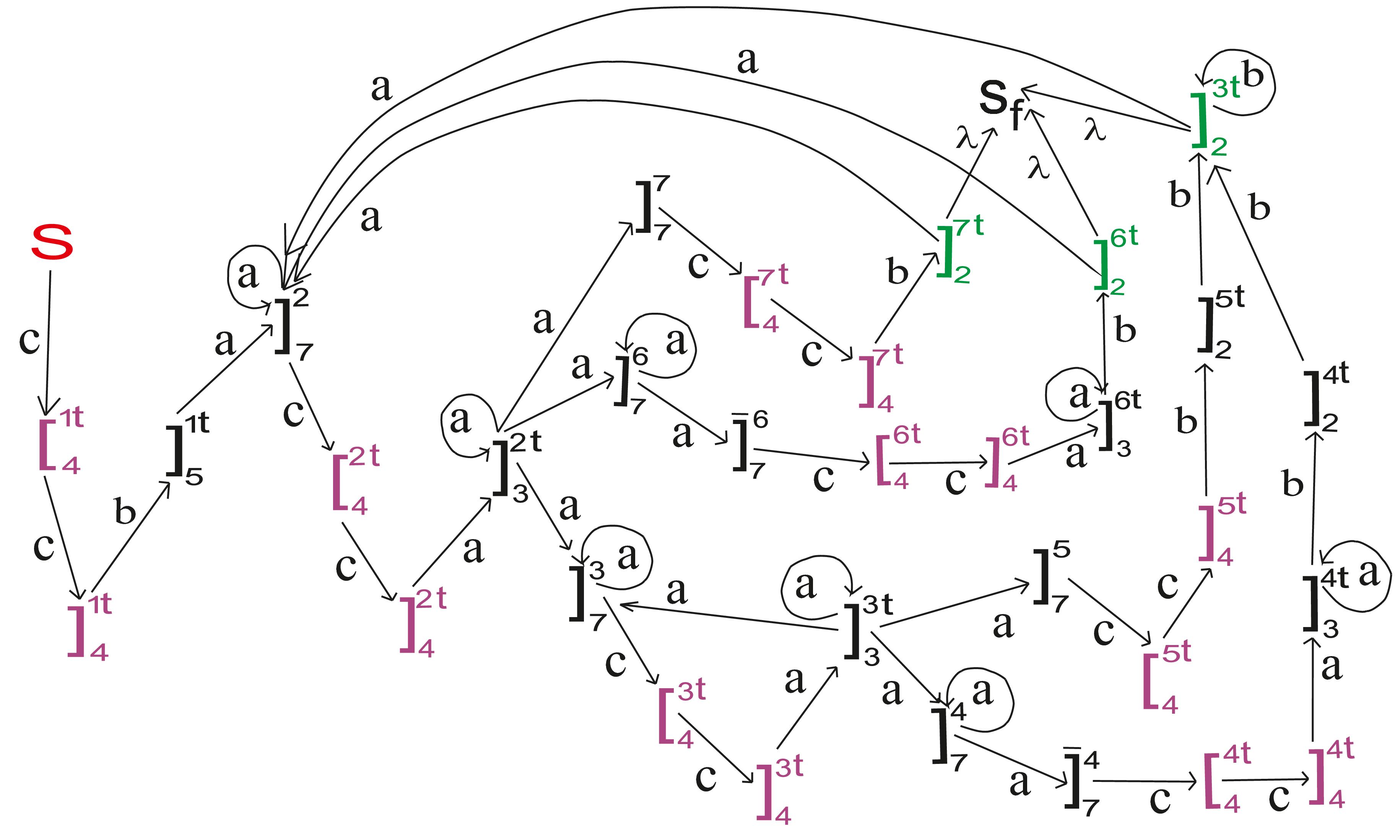}
		\vspace{-0.2cm}
		\caption{{\footnotesize The transition diagram ${\cal A}_e$ built from ${\cal G}.e^S$ in Example 4.2. Each 
		bracket $[_i$ ($S$, $]_i$) in ${\cal A}_e$ corresponds to the state $s_{[_i}$ ($s_S$, $s_{]_i}$) (see 
		Example 5.1 b.). $S$ is the initial vertex, vertices colored in green lead to the final state.}}
	%\label{fig21}
\end{figure}
\vspace{-0.1cm}

\begin{example}
a. The regular grammar that generates the regular superset approximation of the linear language in Example 3.6  is  $G_r\hspace{-0.1cm}=(\{S,]_1,]^t_2,]^t_3,]_4,]^t_5,]_6,[^t_7,]^t_7\}, \{ a,b,c,d\}, S, P_r)$, where\footnote{Note that, 
since there is only one dependency graph that yields only one plus-height regular expression there is no need of the 
labeling procedure described in Section 4.} $P\hspace{-0.1cm}=\{S\hspace{-0.1cm}\rightarrow a]_1,  
]_1\hspace{-0.1cm}\rightarrow b]_4,[_4\rightarrow b]_6, ]_6\hspace{-0.1cm}\rightarrow a]_1/a[^t_7, 
[^t_7 \rightarrow a]^t_7, ]^t_7 \hspace{-0.1cm}\rightarrow d]^t_5, ]^t_5 \hspace{-0.1cm}\rightarrow c]^t_3, 
]^t_3 \hspace{-0.1cm}\rightarrow b]^t_2, ]^t_2 \hspace{-0.1cm}\rightarrow c]^t_3, 
]^t_2 \hspace{-0.1cm}\rightarrow d]^t_5, ]^t_2 \hspace{-0.1cm}\rightarrow \lambda\}$. 
The language generated by $G_r$ is $L(G_r)\hspace{-0.1cm}=\{(abb)^maa(d(cb)^n)^p|n,m,p\geq 1\}=(abb)^+aa(d(cb)^+)^+\hspace{-0.2cm}=h(R)$. 
The transition diagram associated with the finite automaton that accepts $L(G_r)$ is sketched in Figure 1.c.   

\hspace{-0.2cm}b. The regular grammar that generates the regular superset approximation of the context-free  
language in Examples 3.7 and 4.2  is  
$G_r\hspace{-0.1cm}=(\{S,]^{3t}_2,..., ]^{7t}_2, ]^{2t}_3, ]^{3t}_3,]^{4t}_3, ]^{6t}_3,  
[^{1t}_4,..., [^{7t}_4, ]^{1t}_4,...,]^{7t}_4,\\]^{1t}_5, ]^2_7,...,
]^7_7, \bar ]^4_7, \bar ]^6_7\}, \{ a,b,c\}, S, P_r)$, where $P_r\hspace{-0.1cm}=\{ S\hspace{-0.1cm}\rightarrow c[^{1t}_4, 
[^{it}_4 \rightarrow c]^{it}_4, ]^{1t}_4 \hspace{-0.1cm}\rightarrow b]^{1t}_5, ]^{jt}_4 \hspace{-0.1cm}\rightarrow a]^{jt}_3, 
]^{mt}_4 \hspace{-0.1cm}\rightarrow b]^{mt}_2, ]^{1t}_5\hspace{-0.1cm}\rightarrow a]^2_7, ]^j_7\hspace{-0.1cm}\rightarrow a]^j_7, 
]^n_7 \rightarrow c[^{nt}_4, ]^k_7 \hspace{-0.1cm}\rightarrow a\bar ]^k_7, 
\bar ]^k_7 \rightarrow c[^{kt}_4, ]^{2t}_3 \rightarrow a]^3_7/a]^6_7/a]^7_7, ]^{jt}_3 \hspace{-0.1cm}\rightarrow a]^{jt}_3, 
]^{3t}_3 \hspace{-0.1cm}\rightarrow a]^3_7/a]^4_7/a]^5_7, ]^{kt}_3 \hspace{-0.1cm}\rightarrow b]^{kt}_2, ]^{lt}_2 \rightarrow a]^2_7/\lambda , 
]^{ht}_2 \rightarrow b]^{3t}_2, ]^{3t}_2 \rightarrow b]^{3t}_2/a]^2_7 /\lambda | h\in \{4,5\}, i\in \{1,2,3,4,5,6,7\},\\ j\in \{2,3,4,6\}, k\in \{4,6\}, l\in \{6,7\}, m\in \{5,7\}, n\in \{2,3,5,7\}\}$. The transition diagram associated 
with the finite automaton that accepts $L(G_r)$ is sketched in Figure 5.   

%\begin{figure}
%\vspace*{-0.5cm}
%	\centering
%		\includegraphics[scale=0.60]{Figures/ex223.pdf}
%		\vspace{-0.2cm}
%		\caption{a. The transition diagram ${\cal A}_e$, built from ${\cal G}.e^S$, in Figure 2.b. b. The transition 
 %         diagram ${\cal A}_e$, built from ${\cal G}.e^S$, in Figure 2.e. In both graphs, each bracket $[_i$ ($S$, $]_i$) 
	%      corresponds to the state $s_{[_i}$ ($s_S$, $s_{]_i}$, respectively).}
	%\label{fig21}
%\end{figure}

%b. - The regular grammar that generates the regular superset approximation of the context-free  language in Example 4.7  is 
%$G'_r=(\{S,]^t_2,]^t_3,[^t_4, ]^t_4,]^t_5,]_8\}, \{ a,b,c,d\}, S, P_r)$, with $P=\{ S\rightarrow c[^t_4,  
%[^t_4\rightarrow d]^t_4,]^t_4\rightarrow b]^t_2/a]^t_3/b]^t_5,]^t_2\rightarrow b]^t_2/c[^t_4/a]_8, 
%]^t_3\rightarrow a]^t_3/b]^t_2/c[^t_4/a]_8, ]^t_5\rightarrow a]_8, ]_8 \rightarrow a]_8/c[^t_4, 
%]^t_2 \rightarrow \lambda\}$.

%The language generated by $G'_r$ is $L(G'_r)=((cda^*b^+a^*)^*(cda^2a^*)^*)^*(cda^*b^+)^+\hspace{-0.1cm} 
%=\hspace{-0.1cm}[cd(a^*b^++a^2)a^*]^*(cda^*b^+)^+$.
%The transition diagram ${\cal A}_e$ is sketched in Figure 6.
\end{example}
\vspace*{-0.3cm}

\section{Conclusions}

In this paper we have introduced a normal form for context-free grammars, called \textit{Dyck normal form}. Based 
on  this normal form and on graphical approaches we gave an alternative proof of the \textit{Chomsky-Sch$\ddot{u}$tzenberger 
theorem}. From a \textit{transition-like diagram} for a context-free grammar in Dyck normal form we built a \textit{transition diagram} 
for a \textit{finite automaton} and a \textit{regular grammar} for a \textit{regular superset approximation} of the original context-free 
language. A challenging problem for further investigations may be to further refine this superset approximation depending on the type of 
the grammar (e.g. nonself-embedding or unambiguous) or on the size of the grammar (e.g. number of nonterminals, rules, etc.) generating a 
certain context-free language. 

The method used throughout this paper is graphically constructive, and it shows that $i.$ derivational structures in context-free  
grammars can be better described through nested systems of parenthesis (Dyck languages), and $ii.$ the Chomsky-Sch\"utzenberger theorem 
may render a good and efficient approximation for context-free  languages. Furthermore, the method provides a graphical framework to 
handle derivations and descriptional structures in context-free  grammars, which may be useful in further complexity investigations 
of context-free  languages.     

\renewcommand\refname{References}

\end{document}